\documentclass[12pt]{article}

\usepackage{amsthm}
\usepackage{amsfonts}
\usepackage{amsmath}
\usepackage{amssymb}
\usepackage{dsfont}
\usepackage{mathcmd}
\usepackage{enumerate}
\usepackage{graphicx}
\usepackage[page, titletoc]{appendix}
\newcommand{\pin}{\ensuremath{\pi_n}}
\newcommand{\pinx}[1]{\ensuremath{\pi_n(#1)}}

\newcommand{\mut}[1]{\ensuremath{\mu_n^{#1}}}
\newcommand{\mutx}[2]{\ensuremath{\mu_n^{#1}(#2)}}
\newcommand{\dist}[2]{\ensuremath{d_{\text{\tiny TV}}\left(#1,#2\right)}}
\newcommand{\prob}[1]{\ensuremath{\mathbb{P}\left(#1\right)}}
\newcommand{\mean}[1]{\ensuremath{\mathbb{E}\left[#1\right]}}
\newcommand{\standev}[1]{\ensuremath{\sigma\left[#1\right]}}
\newcommand{\ztheta}[1]{\ensuremath{\zeta_n^{#1}}}
\newcommand{\dn}{\ensuremath{\delta_n}}

\newcommand{\ind}[1]{\ensuremath{\mathds{1}_{\{#1\}}}}
\newcommand{\Omgn}{\ensuremath{\Omega_n}}
\newcommand{\Ant}{\ensuremath{A_{n,\theta}}}
\newcommand{\kinf}{\ensuremath{\frac{1}{2}\sqrt{\frac{n}{1-\beta}}}}
\newcommand{\ksup}{\ensuremath{\frac{n}{2}}}
\newcommand{\ksuplog}{\ensuremath{\frac{n}{\log n}}}
\newcommand{\ksuploglog}{\ensuremath{\frac{\sqrt{n}}{\log\log n}}}
\newcommand{\kinflog}{\ensuremath{\frac{1}{2}\log n\sqrt{\frac{n}{1-\beta}}}}
\newcommand{\ksupk}{\ensuremath{{\frac{n}{2}-k}}}
\newcommand{\nhpk}{\ensuremath{\frac{n}{2}+k}}
\newcommand{\fone}{\ensuremath{1+e^{\frac{4\beta k}{n}}}}

\newcommand{\myh}{\ensuremath{h}}
\newcommand{\myp}{\ensuremath{\phi}}
\newcommand{\myH}{\ensuremath{H}}
\newcommand{\myP}{\ensuremath{\Phi}}
\newtheorem{thm}{Theorem}[section]
\newtheorem{lmm}[thm]{Lemma}

\newtheorem{crl}[thm]{Corollary}
\theoremstyle{definition}
\newtheorem{rmk}{Remark}[section]
\newtheorem{dfn}{Definition}[section]
\numberwithin{equation}{section}
\usepackage[affil-it]{authblk}
\usepackage[pdftitle={Entropy-driven cutoff phenomena}]{hyperref}

\begin{document}

\title{Entropy-driven cutoff phenomena}


\author[1]{Carlo Lancia}
\author[2]{Francesca R. Nardi}
\author[1]{Benedetto Scoppola}
\affil[1]{University of Rome TorVergata}
\affil[2]{Technische Universiteit Eindhoven}

\maketitle

\begin{abstract}
In this paper we present, in the context of Diaconis' paradigm, a
general method to detect the cutoff phenomenon. We use this method to
prove cutoff in a variety of models,
some already known and others not yet appeared in literature, including a
chain which is non-reversible w.r.t.\ its stationary measure.
All the given examples clearly indicate that a drift towards the opportune
quantiles of the stationary measure could be held responsible for this
phenomenon. 
In the case of birth-and-death chains this mechanism is fairly well
understood; our work is an effort to generalize this picture to more
general systems, such as systems having stationary measure spread over the
whole state space or systems in which the study of the cutoff may not be 
reduced to a one-dimensional problem.
In those situations the drift may be looked for by means of a suitable
partitioning of the state space into classes; using a statistical mechanics
language it is then possible to set up a kind of energy-entropy
competition between the weight and the size of the classes. Under the lens
of this partitioning one can focus the mentioned drift and prove cutoff
with relative ease.
\end{abstract}

\section{Introduction and Main Results}
\label{sec:intro}

In this paper we present sufficient conditions for a family of finite
ergodic Markov chains to exhibit cutoff. Roughly speaking the cutoff
phenomenon is an abrupt convergence of a Markov chain to its
equilibrium distribution. 
The detailed description of the cutoff phenomenon is given by means of
two quantities, 
the \textit{cutoff-time} and the \textit{cutoff-window}, the latter being much smaller
than the former. For an overview on the cutoff phenomenon we refer the
reader to the review paper by Diaconis~\cite{diaconis1996cpf} and the
book by Levin, Peres and Wilmer~\cite{levin2006mca}.

Our main results, Theorem~\ref{thm2} and its corollary, 
identify with much clarity the cutoff-time as the expected value of
a certain hitting time, and
for the first time in literature such hitting time is related to some entropy
considerations, see Section~\ref{sec:times} below.~Corollary~\ref{crl1}
also gives evidence of the nature of the
cutoff-window, which is in turn kindred to the standard deviation of
the hitting time mentioned above and/or to the mixing features of the
chain.
The level of generality of the key results
gives the possibility to use statistical-mechanics-based ideas to prove
cutoff for a variety of models known in literature, such as Coupon
Collector, Top-in-at-random, Ehrenfest Urn, Random walk on the
hypercube and mean-field Ising model. Furthermore, we prove cutoff for
a couple of one-parameter 
families of random walks, partially biased (i.e.\ with drift) and partially diffusive,
whose peculiar feature is to have cutoff-window of different order
depending on the parameter.
It is worthy to notice that the first of those families is an example
of non-reversible chain exhibiting cutoff (see Section~\ref{sec:sutr-model}).

Section~\ref{sec:framework-notation} defines the structure of our
study, Section~\ref{sec:times} gives some of the ideas standing
behind the main results and draw a comparison with previous
approaches. Section~\ref{sec:key-thm} states our key theorems 
while Section~\ref{sec:discussion} examines
them and gives an explanation of the hypothesis. All the proofs are deferred to
Section~\ref{sec:proof}.  In Section~\ref{sec:apps} we discuss the
application of our results to the models mentioned above.

\subsection{Framework and notation}
\label{sec:framework-notation}

In what follows we will consider \textit{families of finite ergodic Markov chains},
that is sextets of the form 
\[\{\Omgn, X_n^t, P_n, \pi_n, \mut{t}, \mut{0}\}\]
where $\Omgn$ is the finite state space of the $n$-th chain,
$X_n^t$, which has transition matrix $P_n$ and unique stationary
measure $\pi_n$. The symbols $\mut{0}$ and $\mut{t}$ stand for the
initial distribution of the $n$-th chain and its probability
distribution after $t$ steps. The time $t$ is a discrete quantity.
For brevity we will refer to such families simply as families of
Markov chains, omitting the expression ``finite ergodic'' throughout
the whole paper.

\begin{dfn}
\label{def:diac}
  A family of Markov chains is said to exhibit cutoff if there exist
  two sequences of integers, $\{a_n\}$ and $\{b_n\}$ such that
  \begin{equation}
    \frac{b_n}{a_n} \TendsTo[n,\infty] 0\label{eq:10}
  \end{equation}
  and
  \begin{align}
    \lim_{\theta\to\infty} \liminf_{n\to\infty} \dist{\mut{a_n-\theta
        b_n}}{\pin} &= 1 \label{eq:14}\\
    \lim_{\theta\to\infty} \limsup_{n\to\infty} \dist{\mut{a_n+\theta
        b_n}}{\pin} &= 0\label{eq:15}
  \end{align}
\end{dfn}

\begin{figure}[htbp]
\begin{center}
\includegraphics[width=0.9\textwidth]{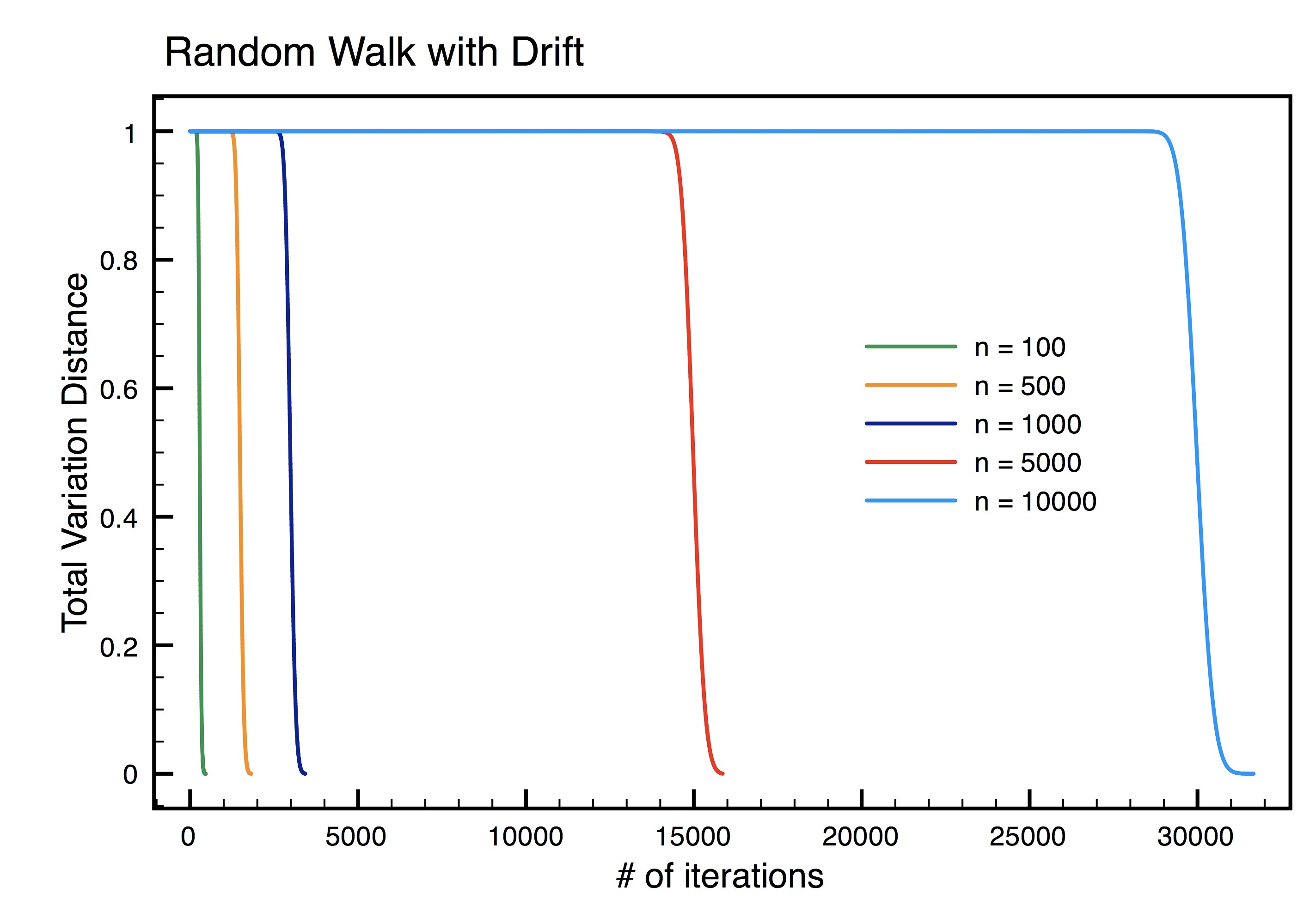}
\caption{Biased random walk on a segment. The transition probabilities
are $P_{i,i-1} = \frac{1}{6}$, $P_{i,i} = \frac{1}{3}$ and $P_{i,i+1} =
\frac{1}{2}$. The curves refer to different values of $n$, the length of the segment.}
\label{fig:brw}
\end{center}
\end{figure}

Equations~\eqref{eq:14} and~\eqref{eq:15} represent the sharp
convergence to the equilibrium distribution, see Figure~\ref{fig:brw}. The distance from
stationarity is taken here to be the usual \textit{total variation distance}
\begin{align}
  \dist{\mut{t}}{\pin} &= \frac{1}{2} \sum_{i\in\Omgn} \left\vert
    \pinx{i} - \mutx{t}{i}\right\vert
  = \max_{A\in\Omgn} \left[\pinx{A} - \mutx{t}{A}\right]\label{eq:23}
\end{align}
\begin{rmk}
  Definition~\ref{def:diac} was first introduced in~\cite{aldous1986sca}.
  Although there exist equivalent alternative definitions of cutoff
  (see~\cite{barrera2009abrupt},~\cite{ding2010total} and~\cite{martinez2001decaying}) we
  prefer to work with the one given, for it leaves us control on the cutoff-window.
\end{rmk}

As mentioned above there exists a connection between the cutoff-time
and the expectation of an hitting time. That connection can be easily
pointed out if we think 
of the total variation distance between \mut{t}{} and \pin{} 
(which, in principle, could be
computed at any given time) as a random variable, or better as a
deterministic object computed at a stochastic time. This idea
motivates the following
\begin{dfn}
  Given a random variable $\xi$, we define the total
  variation distance at time $\xi$ as the following r.v.
  \begin{equation}
    \label{eq:18}
    \dist{\mut{\xi}}{\pin} = \sum_{t \in\mathbb{Z}} \dist{\mut{t_+}}{\pin}
    \ind{\xi = t}
  \end{equation}
  where $t_+ = \max\{0, t\}$. When $\xi$ takes values in
  $[-a,+\infty)$, with $a\in\mathbb{R}^+$, this definition is
  equivalent to
  \begin{align*}
    \dist{\mut{\xi}}{\pin} &= \sum_{t \geq -a} \left[\dist{\mut{t}}{\pin}
    \ind{\xi = t, \, \xi\geq0 }  + \dist{\mut{0}}{\pin}
    \ind{\xi = t, \, \xi < 0 }\right] \tag{\ref{eq:18}a} \label{eq:18abc}
  \end{align*}
\end{dfn}

We need this definition because in the statements of the key theorems
we will consider the expectation of~\eqref{eq:18abc} 
at the stochastic time $\zeta-a$, where $\zeta \geq0$ is a hitting time. 
This is a natural consequence of our aim to care about
the cutoff-window.
The expectation of~\eqref{eq:18abc} can be computed as
\begin{align}
  \mean{\dist{\mut{\xi}}{\pin}} &= \sum_{t \geq 0} \dist{\mut{t}}{\pin}
    \prob{\xi = t, \, \xi\geq0 } + \dist{\mut{0}}{\pin}
    \prob{\xi < 0 } \label{eq:19}
\end{align}
Although the condition $\xi \geq 0$ could be dropped in the first sum
of~\eqref{eq:19}, we prefer to keep it for notational purposes that
will become clear in the proof of Theorem~\ref{thm2}.

\subsection{Cutoff-times and hitting-times}
\label{sec:times}

Theorem~\ref{thm2} and its Corollary~\ref{crl1} bring to light the
link between the cutoff 
phenomenon and the hitting of the \textit{relevant part} of the state
space \Omgn. Relevant part means the subset of the state
space where the stationary distribution \pin{} is mostly concentrated,
see equation~\eqref{eq:26} below.
This seems quite a natural approach when we realize that nearly every
chain known to exhibit cutoff hits the relevant part of the state
space in a quasi-deterministic way, that is the hitting time $\tau_n$
of such a relevant part satisfies the following limit:
\begin{equation}
  \label{eq:27}
  \frac{\standev{\tau_n}}{\mean{\tau_n}} \TendsTo[n,\infty] 0
\end{equation}
where $\standev{\tau_n}$ is the standard deviation of $\tau_n$. 
It is relatively easy to prove a limit as in~\eqref{eq:27} whenever the
chain presents a \textit{drift} towards the relevant part of the state space.
In Section~\ref{sec:apps} we
present a rich selection of examples of applications of our theorems
as well as a comparison with the existing literature. 

The picture of a quasi-deterministic hitting we have described so far
holds as well for the systems with uniform stationary
measure, for which the relevant part of the state space would be
\Omgn{} itself.
As a matter of fact, if we desist from the whole description of such a
chain and look into a suitable projection, then we may find that the original
stationary distribution is no longer uniform. 
The projected stationary distribution, $\nu_n(x)$, is
indeed proportional to the number of states $i \in \Omgn$ which
correspond to $x$ according to the equivalence relation we used to
project the original chain. 
Consequently, since $-\nu_n(x)\log
\nu_n(x)$ is the contribution to the  entropy of $\nu_n$ given by the
$x$-th equivalence
class, we have that the relevant part of the state space is composed of the
classes providing the leading contribution to the entropy. In these
cases the drift mentioned above is therefore supplied by entropic considerations;
we will return to this point later on in Section~\ref{sec:apps}. 
With respect to what we have said above, Corollary~\ref{crl1}
represents then a possible \textit{trait d'union} between two
classes of Markov chains exhibiting cutoff: the
first being made up of chains having stationary measure
concentrated in a small subset of the state
space, like birth-and-death chains with drift, and the second 
composed of those chains with stationary measure uniform or spread over
\Omgn{}, like the random walk on the hypercube, many
card-shuffling models and some high-temperature statistical mechanics models.

The idea of relating cutoff with the hitting of the appropriate quantiles of the
stationary distribution is already present in literature,
see~\cite{ding2010total},~\cite{diaconis2006separation},~\cite{martinez2001decaying}
and~\cite{barrera2009abrupt}.
In~\cite{ding2010total}~and~\cite{diaconis2006separation} the cutoff
is completely characterized for the special case of birth-and-death
chains, respectively in total variation and in separation distance. A
discussion of the results in~\cite{ding2010total} is deferred to
Section~\ref{sec:discussion} after we have stated our main theorems.
With respect to~\cite{martinez2001decaying}
and~\cite{barrera2009abrupt} our approach allows the
study of the cutoff phenomena in a context closer to the classical
Diaconis' paradigm. In particular, with respect to the former
reference we define cutoff in a finite configurations space and consequently
we have a precise control of the cutoff-window. With respect to the
latter, we will show in Sections~\ref{sec:discussion} and~\ref{sec:if-model} that our tackle
to the problem produces a clearer understanding of the role of the
drift in the cutoff phenomena.

\subsection{Key results}
\label{sec:key-thm}

In this first theorem, that will be the main ingredient of the proof of
Theorem~\ref{thm2}, we relate the cutoff phenomenon to systems having
an abrupt convergence to equilibrium at a stochastic time which is
quasi-deterministic in the sense of~\eqref{eq:27}.

\begin{thm}
\label{lmm1}
  Let $\{\Omega_n, X_n^t, P_n, \pi_n, \mut{t}, \mut{0}\}$ a family of
  Markov chains, $\{\tau_n\}$ a family of non-negative
  random variables with finite expected value $E_n = \mean{\tau_n}$ and
  standard deviation $\sigma_n = \standev{\tau_n}$ such that
  \begin{equation}
    \label{eq:1}
    \lim_{n\to\infty} \frac{\sigma_n}{E_n} = 0
  \end{equation}
  Let $\{\delta_n\}$ be a sequence of positive numbers such that
  \begin{equation}
    \lim_{n\to\infty} \frac{\delta_n}{E_n} = 0\label{eq:2}
  \end{equation}
  \begin{align}
    \text{definitively for } n \to \infty & \quad
    \mean{\dist{\mut{\tau_n-\theta\delta_n}}{\pin}} \geq
    1-f(\theta) \label{eq:3}\\
    \text{definitively for } n \to \infty & \quad
    \mean{\dist{\mut{\tau_n+\theta\delta_n}}{\pin}} \leq g(\theta) \label{eq:4}
  \end{align}
  where $f$ and $g$ are two functions tending to $0$ as $\theta \to
  \infty$.

  \noindent Then the family exhibits cutoff with
  \begin{align}
    a_n &= E_n\label{eq:5}\\
    b_n &= O(\sigma_n + \delta_n)
  \end{align}
\end{thm}

\mbox{}

Before we move to the statement of Theorem~\ref{thm2} we need to
introduce some tools.

\begin{dfn}
  We define a \textit{family of nested subsets} as a sequence
  $\{\Ant\}_{\theta\geq1}$ with the following properties: $ \forall \,
  \theta \in\mathbb{N}, \:\: \exists \, N>0 \text{ such that }
  \forall\, n \geq N$
  \begin{align}
    &A_{n,\theta^\prime} \subset \Omgn & \forall \, 1\leq
    \theta^\prime \leq \theta\label{eq:24} \\
    &A_{n,\theta^\prime} \subseteq A_{n,\theta^{\prime\prime}}
    &\forall \, 1\leq \theta^\prime \leq \theta^{\prime\prime}
    \leq\theta\label{eq:25}
  \end{align}
\end{dfn}

\begin{dfn}
  Given a family of nested subsets we shall say that \pin{} is
  \textit{h-concentrated} on \Ant{} if there exists a function
  $h(\theta)$ tending to zero as $\theta \to \infty$ such that
  \begin{equation}
    \label{eq:26}
    \text{definitively as } n\to \infty
    \qquad \pin(\Ant^\complement) < h(\theta)
  \end{equation}
  where $\Ant^\complement = \Omgn \setminus \Ant$.
\end{dfn}


Finally, define $\ztheta{\theta} = \min\{t\geq0 \,:\, X_n^t \in
\Ant\}$ the hitting time of \Ant{}; note that
$\ztheta{\theta} \geq \ztheta{\theta^\prime}$ if $\theta \leq \theta^\prime$.
We are now ready to state the main result of this paper.
\begin{thm}
\label{thm2}
  Let $\{\Omega_n, X_n^t, P_n, \pi_n, \mut{t}, \mut{0}\}$ a family of
  Markov chains. Suppose that $\mut{0}$ is such that there exists a family of
  nested subsets $\{\Ant\}_{\theta\geq1} \subset \Omgn$ with the
  following properties:
  \begin{align}
    & \pin \text{ is h-concentrated in } \Ant\label{eq:6}\\
    &\frac{\standev{\ztheta{1}}}{\mean{\ztheta{1}}} \TendsTo[n,\infty]
    0\label{eq:7}\\
    & \standev{\ztheta{\theta}} \leq \standev{\ztheta{1}}\label{eq:13}
  \end{align}
  and there exists a sequence of positive integers $\{\Delta_n\}$ such that
  \begin{align}
    &\frac{\Delta_n}{\mean{\ztheta{1}}} \TendsTo[n,\infty] 0\label{eq:8}\\
    &\lim_{\theta\to\infty} \lim_{n\to\infty} \frac{\mean{\ztheta{1} -
        \ztheta{\theta}}}{\theta\Delta_n} = 0\label{eq:9}
  \end{align}
  Then there exists a function $f(\theta)$, 
  tending to 0 as $\theta\to\infty$, such that
  \begin{align}
    \mean{\dist{\mut{\ztheta{1}-\theta\delta_n}}{\pin}} \geq
    1-f(\theta) & \quad \text{definitively for } n \to \infty\label{eq:11}
  \end{align}
  where
  \begin{equation}
    \label{eq:12}
    \delta_n = 2\,(\Delta_n + \standev{\ztheta{1}})
  \end{equation}
\end{thm}

\mbox{}

A relatively easy consequence of Theorem~\ref{thm2} is the following

\begin{crl}
\label{crl1}
  Assume that all the hypothesis of Theorem~\ref{thm2} hold for a
  given family of Markov chains.
  In addition suppose that given two copies of the $n$-th chain of the
  family, $Z_n^t$ and $W_n^t$, there exists a coupling $(Z_n^t,
  W_n^t)$ such that
  \begin{align}
   & Z_n^0 = z_0 \in A_{n,1} \quad W_n^0 \sim \pin\label{eq:17}\\
   & \text{if } Z_n^{s^*} = W_n^{s^*} \text{ then } Z_n^s = W_n^s
   \quad \forall\: s\geq s^*\label{eq:70} \\
   & \gamma_n = \min\{t \geq 0 \::\: Z_n^t = W_n^t\} \text{ is such that
   } \nonumber\\
   &\max_{z_0 \in A_{n,1}}\prob{\gamma_n > \theta\delta_n \,|\, Z_n^0 = z_0} \leq
   g(\theta) \text{ definitively as } n\to\infty \label{eq:92}
  \end{align}
  with $g(\theta) \TendsTo[\theta,\infty] 0$.
  Then the family exhibits cutoff with
  \begin{align}
    a_n &= \mean{\ztheta{1}}\label{eq:29}\\
    b_n &= O(\delta_n)\label{eq:30}
  \end{align}
\end{crl}

\subsection{Discussion of Theorem~\ref{thm2}}
\label{sec:discussion}

Theorem~\ref{thm2} identifies a general structure that underlies a
class of systems exhibiting cutoff: those with stationary measure
concentrated in a small region of the state space (\Ant{} in the
theorem, see \eqref{eq:24}-\eqref{eq:26} above).
Although widely general, Theorem~\ref{thm2} is most useful
when we face a family of Markov chains $X_n^t$ which is, or can be
projected onto, a family of birth-and-death chains. In those cases we have
indeed closed formulas to deal
with expectation and variance of the various hitting times, see for
example~\cite{barrera2009abrupt} or~\cite{feller1968intro}. 
The non-reversible random walk on a cylindrical lattice, presented in
Section~\ref{sec:sutr-model}, shows that the application of
Theorem~\ref{thm2} is not restricted solely to those models where the study of
the cutoff can be completely reduced to a one-dimensional problem.


Total variation cutoff was completely settled in~\cite{ding2010total}
for the class of birth-and-death chains, in particular it is shown
therein that we have cutoff if and only if $t^{(n)}_{\text{REL}} =o(t^{(n)}_{\text{MIX}})$, 
where $t^{(n)}_{\text{REL}}$ and $t^{(n)}_{\text{MIX}}$ are
respectively the relaxation time and the mixing
time of the $n$-th chain. It should be pointed out, however, that in
some importat models of statistical mechanics, namely the Ehrenfest
urn and the magnetization chain for the mean field Ising model, a non optimal
$\sqrt{t^{(n)}_{\text{REL}}\cdot t^{(n)}_{\text{MIX}}}$ window
order is found. 
Our approach conversely, provided a suitable definition of the
\Ant{}'s (see Remark~\ref{rmk:1.2} below),
is always capable of delivering the right cutoff-window
order. Moreover, in most situations the computation of $\mean{\ztheta{\theta}}$ and
$\standev{\ztheta{\theta}}$ happens to be less challenging than the 
computation of the spectral gap of the chain.

Within the framework of birth-and-death chains, \pin{} being
concentrated in \Ant{} is equivalent to a drift
of the chain towards \Ant{} itself; such a drift is likely to ensure
\begin{equation}
  \label{eq:106}
  \frac{\standev{\ztheta{\theta}}}{\mean{\ztheta{\theta}}}
  \TendsTo[n,\infty] 0
  \qquad \forall \, \theta \geq 1
\end{equation}
Limit~\eqref{eq:106} means in turn that for $n$ sufficiently large 
the chain will hit \Ant{}
in a quasi-deterministic way, that is the probability of $X_n^t$
being into \Ant{} will suddenly rise from 0 to 1 in a window of size
$\standev{\ztheta{\theta}}$ centered on $\mean{\ztheta{\theta}}$.
This means that, if the system was started outside \Ant{}, it is
undergoing the first part of the cutoff curve, i.e.\ it satisfies~\eqref{eq:14}. 
If the system relaxes
inside \Ant{} in a time interval that is comparable with
$\standev{\ztheta{\theta}}$, then we would experience cutoff with a
window of the order of $\standev{\ztheta{\theta}}$. It is also
possible that the time $t_{\text{mix}}$ needed for the system to relax inside \Ant{} is
larger than $\standev{\ztheta{\theta}}$ but smaller than
$\mean{\ztheta{\theta}}$, implying then cutoff with a cutoff window of
the order of $t_{\text{mix}}$. This is the case of the Ehrenfest Urn
and the Random Walk on the Hypercube, which we present in detail
in Section~\ref{sec:ehrenfest}.

The technical problem we had to face in designing Theorem~\ref{thm2}
is the fact that $\mean{\ztheta{\theta}}$ is not a good candidate to
the cutoff-time, $a_n$, being $\theta$-dependent. This is the reason
why we preferred to split the diffusion inside \Ant{} in two parts: the
hitting of $A_{n,1}$, a subset of \Omgn{} such that
$\pin(A_{n,1})$ is non-vanishing in $n$, 
and the diffusion time once $A_{n,1}$ is reached,
see~\eqref{eq:92}.

\begin{rmk}
\label{rmk:1.2}
  There is no universal choice for the family \Ant, multiple
  definitions are possible and each of them affects indirectly the
  size of the cutoff-window. Remark~\ref{rmk:crucial} in
  Section~\ref{sec:if-model} shows a choice for the \Ant's which leads
  to a non optimal cutoff-window. 
  The applications presented in Section~\ref{sec:apps} also suggest
  the key to obtain an optimal cutoff-window: design the 
  family \Ant{} in such a way that the expected travelling time
  $\mean{\ztheta{1}-\ztheta{\theta}}$ is of the same order in $n$ as the
  time $\theta\delta_n$ necessary to achieve equilibrium starting anywhere
  in~$A_{n,1}$ (cfr. Corollary~\ref{crl1}). 
  From the discussion in this section, and in particular from~\eqref{eq:106}, we
  can take an energy-landscape point of view and visualize our system
  as a single well, where the height of the energy landscape 
  in a given point $i$ increases with
  $\pin^{-1}(i)$. Consider for example the Ehrenfest Urn, presented in 
  Section~\ref{sec:ehrenfest}; requiring that $\mean{\ztheta{1}-\ztheta{\theta}} =
  O(\delta_n)$ corresponds to say that, once 
  the chain has reached the \textit{border} of
  \Ant{}, it falls to the bottom of the well (that is $A_{n,1}$)
   in a time which is also sufficient to diffuse inside the well itself.
\end{rmk}

\begin{rmk}
  \label{rmk:harm}
  Note that, in the case of birth-and-death chains,
  hypotheses~\eqref{eq:13} is trivial.
\end{rmk}

\begin{rmk}
  We would like to emphasize that the task of showing the cutoff behavior is usually
  accomplished by means of a coupling argument. In most situations the
  coupling argument needs to be sufficiently fine, since the desired
  estimates are to be performed at times $a_n \pm \theta b_n$, i.e.\
  with two very different time scales involved. In our approach this
  time scale issue  is set loose 
  when we split the study of the cutoff in two phases, namely the hitting of
  $A_{n,1}$ and the subsequent evolution to equilibrium.  We
  will see later on in the applications (Section~\ref{sec:apps}) that within our
  framework only very basic and
  intuitive couplings are demanded.
\end{rmk}

\section{Proof of Main Results}
\label{sec:proof}

In the following we will make intensively use of two
easy and well known facts that are worth of a brief recalling, before
we proceed with the proof of the key results.

\begin{lmm}{(Cantelli's inequality)}\label{lemma2}
Let Y be a random variable with finite mean $\mu$ and finite variance
$\sigma^2$. Then, for any $\theta\geq 0$
  \begin{equation}
  \prob{Y-\mu \geq \theta\sigma} \leq \frac{1}{1+\theta^2}\label{eq:42}
\end{equation}
\end{lmm}

\begin{lmm}\label{lemma1}
  Let $X(t)$ be a discrete Markov chain with finite space state. Then the
  total variation distance from stationarity 
  is a non-increasing sequence as a function of $t$.
\end{lmm}
\noindent A proof of Lemma~\ref{lemma1} may be found in~\cite{jerrum2003csa}
and a proof of Lemma~\ref{lemma2} in~\cite{feller1968introduction}.

\mbox{}

\noindent
Now we can start with the proof of the key results.

\begin{proof}[Proof of Theorem~\ref{lmm1}]
For brevity of notation set $D(t) \equiv
\dist{\mut{t}}{\pin}$ and $\xi \equiv \tau_n-\theta\delta_n$; note
that according with the latter definition $\mean{\xi}-\theta\sigma_n =
a_n-\theta b_n$. Then, using~\eqref{eq:19}
  \begin{align}
    \mean{D(\xi)} &= \sum_{t\geq0} D(t)\prob{\xi=t \,,\, \xi \geq0} +
    D(0)\prob{\xi<0}\label{eq:36}\\
    &\leq \sum_{t\geq0} D(t)\prob{\xi=t \,,\, \xi \geq0} +\prob{\xi<0}\label{eq:38}
  \end{align}

\noindent We can estimate the sum in~\eqref{eq:38} as follows
\begin{align}
    \sum_{t\geq0} D(t)\prob{\xi=t \,,\, \xi \geq0} &\leq
    \sum_{t=0}^{\mean{\xi}-\theta\sigma_n}D(t)\prob{\xi=t
      \,,\, \xi\geq0} + \nonumber\\
    &\qquad\qquad\sum_{t\geq\mean{\xi}-\theta\sigma_n}
    D(t)\prob{\xi=t \,,\, \xi\geq0} \label{eq:39}\\
    & \leq \prob{0 \leq \xi \leq
      \mean{\xi}-\theta\sigma_n} +D(\mean{\xi}-\theta\sigma_n)\label{eq:40}
\end{align}
where from~\eqref{eq:39} to~\eqref{eq:40} we have used
Lemma~\ref{lemma1} to estimate the second sum.

\noindent Substituting equation~\eqref{eq:40} in~\eqref{eq:38} we
obtain
\begin{equation}
  \label{eq:72}
  \mean{D(\xi)} \leq\prob{\tau_n \leq
    \mean{\tau_n}-\theta\sigma_n} +  D(\mean{\xi}-\theta\sigma_n)
\end{equation}
that is
\begin{equation}
  \label{eq:41}
  \mean{\dist{\mut{\tau_n-\theta\delta_n}}{\pin}} \leq
  \dist{\mut{a_n-\theta b_n}}{\pin} + \prob{\tau_n \leq \mean{\tau_n}-\theta\sigma_n}
\end{equation}
Thus, reverting the inequality, in virtue of~\eqref{eq:3} and~\eqref{eq:42} we arrive at
\begin{equation}
  \label{eq:32}
  1 \geq \liminf_{n\to\infty} \dist{\mut{a_n-\theta b_n}}{\pin} \geq
  1-f(\theta)-\frac{1}{1+\theta^2}
\end{equation}

Now set $\eta = \tau_n+\theta\delta_n$ and notice that
$\mean{\eta}+\theta\sigma_n = a_n + \theta b_n$. Then since $\eta \geq
\theta \delta_n$, by~\eqref{eq:19} we get
\begin{align}
  \mean{D(\eta)} &= \sum_{t\geq \theta\delta_n} D(t)\prob{t = \eta}\label{eq:31}\\
  &\geq \sum_{t=\theta\delta_n}^{\mean{\eta}+\theta\sigma_n}
  D(t)\prob{t=\eta}\label{eq:35}\\
  &\geq D(\mean{\eta}+\theta\sigma_n)
  \sum_{t=\theta\delta_n}^{\mean{\eta}+\theta\sigma_n}\prob{t=\eta}\label{eq:43}\\
  & = D(\mean{\eta}+\theta\sigma_n)\prob{\eta \leq
    \mean{\eta}+\theta\sigma_n}\label{eq:44}\\
  & \geq D(\mean{\eta}+\theta\sigma_n)
  \left(1-\frac{1}{1+\theta^2}\right)\label{eq:45}\\
  & \geq D(\mean{\eta}+\theta\sigma_n) -\frac{1}{1+\theta^2}\label{eq:46}
\end{align}
where from~\eqref{eq:44} to~\eqref{eq:45} we used Lemma~\ref{lemma2}.
Reverting the inequality we obtain
\begin{equation}
  \label{eq:47}
  \dist{\mut{a_n+\theta b_n}}{\pin} \leq
  \mean{\dist{\mut{\tau_n+\theta\delta_n}}{\pin}} + \frac{1}{1+\theta^2}
\end{equation}
Therefore, in virtue of~\eqref{eq:4}
\begin{equation}
  \label{eq:48}
  0 \leq \limsup_{n\to\infty} \dist{\mut{a_n+\theta b_n}}{\pin} \leq
  g(\theta) + \frac{1}{1+\theta^2}
\end{equation}

Eventually, mark that~\eqref{eq:1} and~\eqref{eq:2}
infer~\eqref{eq:10}. Passing to the limits for $\theta$ tending to $\infty$
in~\eqref{eq:32} and~\eqref{eq:48} concludes the proof.\qed
\end{proof}

\mbox{}

\begin{proof}[Proof of Theorem~\ref{thm2}]
  Fix $\theta >1$ arbitrarily and consider $n$ sufficiently large to
  ensure~\eqref{eq:24}.
  As in the proof of Theorem~\ref{lmm1} set
  $D(t) = \dist{\mut{t}}{\pin}$ and $\xi = \ztheta{1}-\theta\dn$. By~\eqref{eq:19}
  \begin{align}
    \mean{D(\xi)}&=\sum_{t\geq0} D(t)\prob{\xi=t \,,\, \xi \geq0}
    +D(0)\prob{\xi<0}\label{eq:51}\\
    &\geq \sum_{t\geq0} D(t)\prob{\xi=t \,,\, \xi \geq0}\label{eq:52}\\
    &=\sum_{t\geq0} \prob{\xi=t \:,\: \xi
        \geq0}\frac{1}{2}\sum_{i\in\Omgn}
    \left\vert \mutx{t}{i}-\pinx{i} \right\vert\label{eq:53}\\
    &=\prob{\xi \geq0}\left[\frac{1}{2}\sum_{i\in\Omgn} \sum_{t\geq0}
    \left\vert \left(\mutx{t}{i}-\pinx{i}\right)\prob{\xi=t \:|\: \xi
        \geq0}\right\vert\right]\label{eq:54}\\
    &\geq \prob{\xi \geq0}\left[\frac{1}{2}\sum_{i\in\Omgn} 
    \left\vert \pinx{i}-\sum_{t\geq0}\mutx{t}{i}\prob{\xi=t \:|\: \xi
        \geq0}\right\vert\right]\label{eq:55}
  \end{align}
At this point we note that $\rho_n(i) = \sum_{t\geq0}\mutx{t}{i}\prob{\xi=t \:|\:
  \xi \geq0}$ is a probability distribution on \Omgn{}, for
\begin{equation}
  \label{eq:56}
  \sum_{i\in\Omgn}\rho_n(i) = \sum_{t\geq0}\prob{\xi=t \:|\: \xi
    \geq0}\sum_{i\in\Omgn}\mutx{t}{x}=1
\end{equation}
Hence using~\eqref{eq:23} we have that, for $n$ sufficiently large,
\begin{align}
  \mean{D(\xi)}&\geq \prob{\xi \geq0}
  \max_{A\subseteq\Omgn}\left[\pinx{A} -
    \sum_{t\geq0}\mutx{t}{A}\prob{\xi=t \:|\:\xi
      \geq0}\right]\label{eq:57}\\
  &\geq \prob{\xi \geq0}\left[\pinx{\Ant} -
    \sum_{t\geq0}\mutx{t}{\Ant}\prob{\xi=t \:|\:\xi
      \geq0}\right]\label{eq:58}\\
  &\geq \prob{\xi \geq0}(1-h(\theta)) -
  \sum_{t\geq0}\mutx{t}{\Ant}\prob{\xi=t \,,\,\xi\geq0}\label{eq:61}
\end{align}

We can estimate the first term of the sum in~\eqref{eq:61} by
virtue of Lemma~\ref{lemma2}: 
\begin{align}
  (1-h(\theta))\,\prob{\xi\geq0} &=
  (1-h(\theta))\,\prob{\mean{\ztheta{1}}-\ztheta{1} \leq
    \mean{\ztheta{1}}-\theta\dn}\label{eq:59}\\
  &\geq (1-h(\theta))\left(1-\frac{\text{Var}[\ztheta{1}]}{\text{Var}[\ztheta{1}]
      +(\mean{\ztheta{1}}-\theta\dn)^2}\right)
\end{align}
By~\eqref{eq:7},~\eqref{eq:8} and~\eqref{eq:12} we have that, definitively for $n \to
\infty$, $\prob{\xi\geq0}$ is greater than any
function of $\theta$ tending to one, say $1-\frac{1}{\theta}$. Thus
for $n$ sufficiently large we have that
\begin{equation}
  \label{eq:60}
   (1-h(\theta))\,\prob{\xi\geq0} \geq 1-h(\theta) - \frac{1}{\theta}
\end{equation}
Next consider the remaining term of~\eqref{eq:61}:
\begin{align}
  &\sum_{t\geq0}\mutx{t}{\Ant}\,
  \prob{\ztheta{1}-\theta\dn=t\,,\,\ztheta{1}-\theta\dn\geq0}\label{eq:62}\\
  &\leq \sum_{t\geq0}
  \prob{t\geq\ztheta{\theta}}\,\prob{\ztheta{1}-\theta\dn=t}\label{eq:63}\\
  &\leq\sum_{t=\mean{\ztheta{1}}-\theta\dn-\theta\standev{\ztheta{1}}}^{\mean{\ztheta{1}}-\theta\dn+\theta\standev{\ztheta{1}}}
  \prob{t\geq\ztheta{\theta}}\,\prob{\ztheta{1}-\theta\dn=t} +
  \frac{1}{\theta^2}\label{eq:64}\\
  &\leq \prob{\ztheta{\theta} \leq
    \mean{\ztheta{1}}-\theta\dn+\theta\standev{\ztheta{1}}}+\frac{1}{\theta^2}\label{eq:65}\\
  &=\prob{\mean{\ztheta{\theta}}-\ztheta{\theta} \geq
    2\theta\Delta_n+\theta\standev{\ztheta{1}} -
    \mean{\ztheta{1}-\ztheta{\theta}}}+\frac{1}{\theta^2}\label{eq:67}
\end{align}
Now we have to face possibly two scenarios:
\begin{enumerate}[a.]
\item $\standev{\ztheta{1}} = o(\Delta_n)$
\item $\Delta_n = o(\standev{\ztheta{1}})$ or $\Delta_n = O(\standev{\ztheta{1}})$
\end{enumerate}
In the former case we have that also
$\sigma(\ztheta{\theta})$ is $o(\Delta_n)$ in virtue
of~\eqref{eq:13}. Therefore we can rewrite 
the first term of~\eqref{eq:67} as
\begin{align}
  \prob{\mean{\ztheta{\theta}}-\ztheta{\theta} \geq
    \theta\Delta_n(2+o(1))} &\leq 
  \frac{1}{1+\frac{1}{\sigma^2(\ztheta{\theta})}\,(\theta\Delta_n(2+o(1)))^2}\label{eq:66}\\
  &\leq \frac{1}{1+\theta^2} \text{ definitively as } n\to\infty
\end{align}
In the latter case we have that $\standev{\ztheta{1}}$
satisfies an equation of the kind of~\eqref{eq:9} as well as $\Delta_n$.
Then
\begin{align}
  &\prob{\mean{\ztheta{\theta}}-\ztheta{\theta} \geq
    \theta\standev{\ztheta{1}}\left(1+\frac{\Delta_n}{\standev{\ztheta{1}}} -
      \frac{\mean{\ztheta{1}-\ztheta{\theta}}}{\theta\standev{\ztheta{1}}}\right)}\nonumber\\
  &\leq \frac{1}{1+\left(\frac{\standev{\ztheta{1}}}{\sigma(\ztheta{\theta})}\right)^2\theta^2 \left(1+\frac{\Delta_n}{\standev{\ztheta{1}}} -
      \frac{\mean{\ztheta{1}-\ztheta{\theta}}}{\theta\standev{\ztheta{1}}}\right)^2}\label{eq:68}\\
  &\leq \frac{1}{1+\theta^2 \left(1+\frac{\Delta_n}{\standev{\ztheta{1}}} -
      \frac{\mean{\ztheta{1}-\ztheta{\theta}}}{\theta\standev{\ztheta{1}}}\right)^2}\label{eq:69}
\end{align}
by virtue of~\eqref{eq:13}. 

\noindent Therefore we can infer that for $n$
sufficiently large there exists a function $f(\theta)$ tending to 0 as
$\theta\to\infty$ that satisfies~\eqref{eq:11}.\qed
\end{proof}

\begin{rmk}
  In the proof of next result, Corollary~\ref{crl1}, we will need the
  following equality
  \begin{equation}
    \lim_{M\to\infty} \prob{\ztheta{1} \geq M} = 0 \label{eq:16}
  \end{equation}
  which is an easy consequence of Lemma~\ref{lemma2} and~\eqref{eq:7}.
\end{rmk}

\begin{proof}[Proof of Corollary~\ref{crl1}]
We construct a coupling $(X_n^t,Y_n^t)$ of \mut{t} and \pin{} as follows:
\begin{enumerate}
\item set $X_n^0 \sim \mut{0}$ and $Y_n^0 \sim \pin$, and define
  $\hat{\gamma}_n = \min\{t\geq0 \::\: X_n^t=Y_n^t\}$, first coalescence time
\item for $0\leq t \leq \ztheta{1}$:
  \begin{enumerate}
  \item $X_n^t$ and $Y_n^t$ evolve independently until $\hat{\gamma}_n$, if
    $\hat{\gamma}_n < \ztheta{1}$
  \item $X_n^t = Y_n^t$~~$\forall\, \hat{\gamma}_n \leq t \leq \ztheta{1}$, if any
  \end{enumerate}
\item set $Z_n^0 = X_n^{\ztheta{1}}$ and $W_n^0 = Y_n^{\ztheta{1}}$, then for
  all $t > \ztheta{1}$ run 
  the coupling of $Z_n^t$ and $Y_n^t$ and set $(X_n^t,Y_n^t) = (Z_n^t,
  W_n^t)$.
\end{enumerate}
We have built the coupling $(X_n^t,Y_n^t)$ in this fashion
to have the following property: given that $\ztheta{1} = T <\infty$, for
all $z_0 \in A_{n,1}$
\begin{equation}
\prob{\hat{\gamma}_n > T + \theta\delta_n \,|\, X_n^{T} = z_0} = 
   \prob{\gamma_n > \theta\delta_n \,|\, Z_n^{0} = z_0} \label{eq:179}
\end{equation}
where, according to the notation introduced in Corollary~\ref{crl1}, 
$\gamma_n$ is the first coalescence time of $Z_n^t$ and $W_n^t$.
The idea is then to use the \textit{Coupling Lemma} on the coupling
$(X_n^t,Y_n^t)$ using the informations we already possess from
$(Z_n^t,Y_n^t)$, that is line~\eqref{eq:92}.
So let us take an arbitrary $M$, then
\begin{align}
  \dist{\mut{{\ztheta{1}+\theta\delta_n}}}{\pin} &=\sum_{T\geq0}
  \dist{\mut{{T+\theta\delta_n}}}{\pin} \ind{\ztheta{1} =
    T}\label{eq:97}\\
  & \leq \sum_{T=0}^M \dist{\mut{{T+\theta\delta_n}}}{\pin} \ind{\ztheta{1} =
    T} + \ind{\ztheta{1} \geq M}\label{eq:20}\\
  &\leq\sum_{T=0}^M \prob{\hat{\gamma}_n > T + \theta\delta_n \,|\, X_n^0 = x_0}\ind{\ztheta{1} =
    T} + \ind{\ztheta{1} \geq M}\label{eq:98}\\
  &= \sum_{T=0}^M \sum_{z_0 \in A_{n,1}} \Bigg[\prob{\hat{\gamma}_n > T +
    \theta\delta_n \,|\, X_n^0 = x_0, X_n^T = z_0} \nonumber\\
  & \qquad\qquad \frac{\prob{X_n^0 = x_0, X_n^T = z_0}}{\prob{X_n^0 =
      x_0}}\ind{\ztheta{1} = T} \Bigg] + \ind{\ztheta{1} \geq M}\\
  &\leq \sum_{T=0}^M \max_{z_0 \in A_{n,1}}\prob{\hat{\gamma}_n > T +
    \theta\delta_n \,|\, X_n^T = z_0}\ind{\ztheta{1} = T} + \ind{\ztheta{1} \geq M}\label{eq:126}
\end{align}
By means of ~\eqref{eq:92} and~\eqref{eq:179}
we have that for $n$ sufficiently large
\begin{equation}
  \label{eq:22}
  \dist{\mut{{\ztheta{1}+\theta\delta_n}}}{\pin} \leq g(\theta)
  \ind{\ztheta{1} \leq M} + \ind{\ztheta{1} \geq M}
\end{equation}
Finally, passing to the expectation in~\eqref{eq:22}, by means
of~\eqref{eq:16}, we get
\begin{align}
  \mean{\dist{\mut{{\ztheta{1}+\theta\delta_n}}}{\pin}}
  \leq g(\theta) \text{ definitively as } n\to\infty \label{eq:127}
\end{align}

Indentifying $\tau_n$ with $\ztheta{1}$ we have obtained~\eqref{eq:4} of
Theorem~\ref{lmm1}, while Theorem~\ref{thm2} gives
us~\eqref{eq:1}, the definition of $\delta_n$
via~\eqref{eq:12},~\eqref{eq:2} and~\eqref{eq:3}.
Therefore we have that the family of Markov chains
exhibits cutoff with $a_n=\mean{\ztheta{1}}$ and $b_n = O(2\Delta_n+3\standev{\ztheta{1}}) = O(\delta_n)$. \qed
\end{proof}

\begin{rmk}
\label{rmk:bad}
  Assume now that the state space \Omgn{} is endowed with a
  nearest-neighborhood binary relation. Such a
  relation naturally defines over $\Omgn$ a graph $G(\Omgn, E)$, and
  therefore a metric $d:\Omgn \times \Omgn \to \mathbb{N}$. 
  For any event $A\subseteq\Omgn$ it is
  then reasonable to define the set of the extremal points of $A$ as 
  \begin{equation}
    \label{eq:A1}
    \partial A = \{i \in A \::\: \exists \, j \in \Omgn \setminus A,\, d(i,j)=1 \}
  \end{equation}
  If the family of Markov chains is a nearest-neighbor dynamics, that
  is $P_{ij} = 0$ whenever $d(i,j) > 1$, we know for sure that $X_n^t$
  cannot \textit{jump} inside $A_{n,1}$ but
  is going to hit it on its border, that is
  $X_n^{\ztheta{1}} \in \partial A_{n,1}$. Thus we can ask less than~\eqref{eq:92}
  to the coupling $(Z_n^t,W_n^t)$, specifically
  \begin{equation*}
    \label{eq:174}
    \max_{z_0 \in \partial A_{n,1}} \prob{\gamma_n >
      \theta\delta_n \,|\, Z_n^0 = z_0} < g(\theta) \quad
    \text{definitively as } n \to \infty \tag{\ref{eq:92}a}
  \end{equation*}
  Also, it is not infrequent whatsoever facing Markov chains where the
  state space $\Omgn$ can be put in
  a one-to-one correspondence with a finite subset of $\mathbb{Z}$,
  then the graph $G(\Omgn, E)$ defined above is
  just a discrete segment, and 
  \begin{equation}
    \partial A = \{i \in A \::\:i+1 \not\in A \text{ or } i-1 \not\in
    A\}\tag{\ref{eq:A1}a}\label{eq:90a1}
  \end{equation}
  is composed of just two points. In those situations depending on
  $\mut{0}$ we could be able to determine which
  point of $\partial A_{n,1}$ will be hit by $X_n^t$ so that the
  \emph{max} in~\eqref{eq:174} would not be needed at all.
\end{rmk}

\section{Some Applications}
\label{sec:apps}

\subsection{The Coupon Collector Model}
\label{sec:ccm}

The Coupon Collector Model is a pure-death chain on the state space
$\Omgn=\{0,1,2,\ldots,n\}$, more specifically it is a chain with the following
transition rates:
\begin{equation}
  \label{eq:73}
  q_i = P_{i,i-1} = \frac{i}{n} \quad r_i =P_{i,i} = \frac{n-i}{n} \quad p_i = P_{i,i+1}=0
\end{equation}
This model was introduced in~\cite{erodsreny1961} and it is discussed
in many classical probability books, see e.g.~\cite{levin2006mca} and
references therein. The model can be easily accommodated in our
general framework. We give an alternative description of the
cutoff in this context by means of Theorem~\ref{lmm1}.

The chain clearly has a drift towards the state 0, for it just cannot
move to the right. The equilibrium distribution is $\pi_n =
\delta_{i,0}$, where $\delta_{i,j}$ is the usual Kronecker's delta;
the initial distribution is taken to be $\mut{0} = \delta_{i,n}$.
The hitting time of the state 0 is $\tau_n^0$, which happens to be a 
\textit{strong stationary time}. Thus, we have that for any finite time $t$
\begin{equation}
  \label{eq:74}
  \prob{X_n^t = i \:|\: t \geq \tau_n^0} = \pinx{i}
\end{equation}
Besides, to the leading order $\mean{\tau_n^0}=n\log n$ and
$\standev{\tau_n^0} = n$.

By~\eqref{eq:74}, following the same steps we made from~\eqref{eq:97}
to~\eqref{eq:22}, 
we have that for any $c \geq 0$
\begin{equation}
  \label{eq:75}
  \mean{\dist{\mut{\tau_n^0+c}}{\pin}} = 0
\end{equation}

Next, recall that $D(t)=\dist{\mut{t}}{\pin}$ and take $\xi = \tau_n^0 - 2\theta
n$ and $A=\{0\}$, then from line~\eqref{eq:57} we get
\begin{equation}
  \label{eq:78}
  \mean{D(\xi)}\geq \prob{\xi\geq0} -
  \sum_{t\geq0}\prob{X_n^t=0}\prob{\xi=t,\xi\geq0}
\end{equation}
Now
\begin{align}
  \prob{\xi\geq0}
  &= \prob{n \log n - \tau_n^0 \leq n(\log n -
    2\theta)}\label{eq:80}\\
  &\geq 1-\frac{1}{1+(\log n - 2\theta)^2}\label{eq:81}
\end{align}
and
\begin{align}
  &\sum_{t\geq0}\prob{X_n^t=0}\prob{\xi=t \,,\, \xi \geq0} \nonumber\\
  &\leq \sum_{t=n\log n - 3\theta n}^{n\log n - \theta n}\prob{t \geq
    \tau_n^0}\prob{t=\tau_n^0-2\theta n} +
  \frac{1}{\theta^2}\label{eq:83}\\
  & \leq \prob{n\log n - \tau_n^0  \geq \theta n}\label{eq:84} + \frac{1}{\theta^2}\\
  & \leq \frac{1}{1+\theta^2}+\frac{1}{\theta^2}\label{eq:85}
\end{align}
Thus, for $n$ sufficiently large, there exists a function $f(\theta)$
which tends to 0 as $\theta \to \infty$ such that
\begin{equation}
  \label{eq:86}
  \mean{\dist{\mut{\tau_n^0-\theta n}}{\pin}} \geq 1-f(\theta)
\end{equation}
and by virtue of Theorem~\ref{lmm1} we have that the coupon collector
exhibits cutoff with $a_n = \mean{\tau_n^0} = n\log n$ and $b_n = O(
\standev{\tau_n^0}) = O(n)$.

\subsection{The Top-in-at-random model}
\label{sec:tiar}

The Top-in-at-random is a card shuffling model introduced first
in~\cite{aldous1986sca} and it is the first example in which the
cutoff phenomenon has been recognized.
 The state space
\Omgn{} is the symmetric group, that is the set of all $n!$ possible
permutations of a deck of $n$ cards. The chain describing the model
evolutes according to the following shuffling procedure: pick the
first card of the deck and insert it in the deck at a position chosen
uniformly at random. The equilibrium distribution \pin{} is uniform.
Here we give a description of cutoff in this case using Theorem~\ref{thm2}.

Given the initial permutation $\rho_0$, without loss of generality we
shall imagine to relabel the cards from 1 to $n$, being 1 the bottom
card and $n$ the topmost one.
Next, consider the sets $R_\theta$ composed of those permutations
$\rho$ having the cards from 1 up to $\theta+1$ in crescent relative
order. This corresponds to say that the first \textit{rising sequence} has
length $l \geq \theta+1$, see~\cite{bayer1992tds} for the definition of rising
sequence and for its properties.
To evaluate the cardinality of $R_\theta$ we use the following
argument: given a permutation $\rho \in R_\theta$ keep fixed all the cards
displaying a face value bigger than $\theta+1$ and permute in all
possible ways the remaining.
Call $\mathcal{P}(\rho)$ the set of such permutations, its cardinality
is $(\theta+1)!$ and clearly $\mathcal{P}(\rho) \cap \mathcal{P}(\rho^\prime)
= \emptyset$ if $\rho \not= \rho^\prime$. As $\cup_{\rho\in R_\theta}
\mathcal{P}(\rho) = \Omgn$ we have obtained the following result:
\begin{equation}
  \label{eq:87}
  |R_\theta| = \frac{n!}{(\theta+1)!}
\end{equation}

Please note that $\{\rho_0\} = R_{n-1} \subset R_{n-2} \subset
\cdots\subset R_1 = \Omgn$.
Thus we define the set $\Ant = \Omgn \setminus R_\theta$, that is the
set of all permutations having the first rising sequence of length at
most $\theta$; note that~\eqref{eq:6} is fulfilled.
Define $\ztheta{\theta}$ as the hitting time of \Ant{} and
$\tau_n^\theta$ as the first time when the card $\theta$ reaches the
topmost position; $\tau_n^\theta$ can be restated as the hitting time of $B_{n,\theta}
\subset \Ant$, where $B_{n,\theta}$ is the set of all permutations in \Ant{} having the card
$\theta$ at the topmost position.
Clearly,
\begin{equation}
  \label{eq:99}
  \tau_n^{\theta+1} \leq \ztheta{\theta} \leq \tau_n^\theta \qquad
  \forall \; 1 \leq \theta \leq n-1
\end{equation}
It is easy to find that
\begin{align}
  \mean{\tau_n^\theta} &= n\log n - n\log\theta\label{eq:100}\\
  \text{Var}[\tau_n^\theta] &= \frac{n^2}{\theta} + o(n^2)\label{eq:101}
\end{align}
and therefore
\begin{equation}
  \label{eq:209}
  \mean{\ztheta{1}} = n\log n (1+o(1)) \qquad\qquad
  \mean{\ztheta{1}-\ztheta{\theta}} \leq n\log(\theta+1)
\end{equation}
Moreover, the variances present a property of monotonicity, because
$\forall\,\theta\geq1$ we have that
$\ztheta{\theta}-\tau_n^{\theta+1}$ is independent of
$\tau_n^{\theta+1}$ and $\tau_n^{\theta} - \ztheta{\theta}$ is
independent of $\ztheta{\theta}$.
Therefore,
\begin{align}
  \text{Var}[\tau_n^{\theta+1}] \leq \text{Var}[\ztheta{\theta}] \leq
  \text{Var}[\tau_n^\theta]\label{eq:88} 
\end{align}
Hence to the leading order in $n$,
\begin{align}
  \label{eq:102}
  \mean{\ztheta{1}} &= n\log n\\
  \label{eq:103}
  \standev{\ztheta{1}} &= O(n)
\end{align}

Taking $\Delta_n = n$ we find that all the hypothesis of
Theorem~\ref{thm2} are satisfied. Eventually, $\ztheta{1}$ is a strong
stationary time so that~\eqref{eq:74}-\eqref{eq:75} hold, with
$\tau_n^0$ replaced by $\ztheta{1}$; 
thus via Theorem~\ref{lmm1} the Top-in-at-random model
exhibits cutoff with $a_n = n\log n$ and $b_n = O(n)$.

\subsection{The Ehrenfest Urn model}
\label{sec:ehrenfest}

The Ehrenfest Urn model is possibly the most famous model of
diffusion. The cutoff phenomenon for this chain was first showed
in~\cite{diaconis1990asymptotic}, see also the
review~\cite{diaconis1996cpf} and the 
references therein.

In this model we have two boxes containing a total amount
of $n$ particles, each of them independently change container with probability
$\frac{1}{2n}$. If $X_n^t$ is defined as the number of balls in Urn 1
and that contains $i$ balls then the transition rates for the
Ehrenfest chain are
\begin{equation}
q_i = P_{i,i-1} = \frac{i}{2n} \quad  r_i = P_{i,i} = \frac{1}{2} \quad
p_i = P_{i,i+1} = \frac{n-i}{2n}  \label{eq:89}
\end{equation}
According to~\eqref{eq:89} the Ehrenfest chain is a lazy birth-and-death
chain on $\Omgn = \{0,1,\ldots,n\}$ and its stationary distribution is
a binomial $\mathcal{B}(n,\frac{1}{2})$.

Let us discuss the cutoff-time and the cutoff-window in this case
using the results from Section~\ref{sec:key-thm}.
A good choice for the family of nested subsets is the following:
\begin{equation}
  A_{n,\theta} = \left\{i \in \Omgn \::\:
    \left\vert i-\frac{n}{2}\right\vert\leq
    \frac{\theta}{2}\sqrt{n}\right\}\label{eq:116}
\end{equation}
since $\pi_n(A^\complement_{n,\theta})<\frac{1}{\theta^2}$ by
means of Chebyshev's inequality.
Suppose now that $\mut{0} = \delta_{i,0}$, that is at time 0 Urn 1 is empty;
plain but lenghty calculations (presented for the sake of completeness in
Appendix~\ref{app:mean-var}) 
show that, to the leading order in $n$
\begin{equation}
  \label{eq:90}
  \mean{\ztheta{1}} = \frac{1}{2}n\log n \qquad\qquad
  \mean{\ztheta{1}-\ztheta{\theta}}  = n\log\theta \qquad\qquad
  \standev{\ztheta{1}} = O(n)
\end{equation}
and therefore the hypotheses of Theorem~\ref{thm2} are fulfilled
choosing $\Delta_n = O(n)$ (recall Remark~\ref{rmk:harm}). This last choice sets $\delta_n = O(n)$
and then what we are left with is verifying that Corollary~\ref{crl1} holds.

The Lazy Ehrenfest Urn shares this
feature with the Mean-field Ising model so we defer the matter to
Section~\ref{sec:proof-cutoff} (see in particular
Remark~\ref{rmk:beta}). Eventually, we have proved that the Lazy
Ehrenfest Urn exhibit cutoff with $a_n = \frac{1}{2}n\log n$ and $b_n
= O(n)$.

\subsubsection{The Lazy Random Walk on the Hypercube}
\label{sec:rwh}

In this model the state space is a $n$-dimensional hypercube,
$\Omgn=\{0,1\}^n$; each state can be then represented as a binary $n$-tuple
$x=(x_1,\ldots,x_n)$.
Without loss of generality, let the chain be at time zero at the vertex
$(0,\ldots,0)$, then at each step we flip with probability $\frac{1}{2}$ a
component of the  tuple chosen uniformly at random. This corresponds
to the following update procedure: at each step we choose one of the
possible $n$  directions in space and move along it with probability
$\frac{1}{2}$, while with probability $\frac{1}{2}$ we stand still.
The equilibrium distribution is clearly the uniform one.

The standard treatment of this model is to project it onto a
birth-and-death chain by means of the following equivalence relation:
\begin{equation}
  \label{eq:104}
  x \sim y \qquad \text{iff} \quad \Vert x \Vert_{\ell_1} = \Vert y \Vert_{\ell_1}
\end{equation}
where $\Vert x \Vert_{\ell_1} = \sum_{i} x_i$ is the Hamming weight of
the vertex $x$. The quotient state space $\Omgn / \sim$
can be put into a one-to-one correspondence with the state space
$\Omgn^\sharp=\{0,1,\ldots,n\}$ of a new
chain $X_n^{\sharp,t}$, having transition rates given
by~\eqref{eq:89} and equilibrium distribution equal to a
binomial $\mathcal{B}\left(n,\frac{1}{2}\right)$.

Let us name $\mu_n^{\sharp,t}$ the evolute measure after $t$ steps of
the projected chain $X_n^{\sharp,t}$ and by $\pin^\sharp$ its
equilibrium distribution, then it is a standard task to shown that
\begin{equation}
  \label{eq:105}
  \dist{\mut{t}}{\pin} = \dist{\mu_n^{\sharp,t}}{\pin^\sharp}
\end{equation}
Thus the Lazy Random Walk on the Hypercube exhibits cutoff with the
same cutoff-time and cutoff-window of the Lazy Ehrenfest Urn.

\begin{rmk}
\label{rmk:entropy_driven}
  Since \pin{} is uniform the projected stationary distribution
  $\pin^\sharp(i)$ is clearly proportional to the number of vertices
  having Hamming weight equal to $i$. Therefore $\pin^\sharp$ is
  binomial and is supported in the sense of~\eqref{eq:26} on
  \Ant{}. As the configurations in \Ant{} give the leading
  contribution to the entropy of the distribution $\pin^\sharp$, we
  say that the system is \textit{entropy-driven} towards the
  stationarity. This drift ensures that the conditions of
  Theorem~\ref{thm2} hold although the original distribution on the
  hypercube cannot provide any drift, being uniform.
\end{rmk}

\subsection{Non-reversible biased random walk on a cylinder}
\label{sec:sutr-model}


Consider a family of Markov chains $\{\Omgn, X_n^t, P_n, \pi_n, \mut{t}, \mut{0}\}$ having space state
\begin{equation}
  \label{eq:1s}
  \Omgn = \{(\myh, \myp) \::\: \myh \in \{0,1,\ldots, l-1\}, \:\myp
  \in \{0,1,\ldots,m-1\}\} \qquad \text{with } |\Omgn| = n = l \cdot m
\end{equation}
As stated more precisely below, we are going to regard \Omgn{} as a
cylindrical lattice of volume $n$, having height $l$ and base
circumference of lenght $m$.
The transition kernel of the $n$-th chain is $P_n$, whose entries are given by the following transition probabilities: 
\begin{equation}
  \label{eq:2s}
  \prob{X_n^{t+1} = (\myh^\prime, \myp^\prime) \,|\, X_n^{t} = (\myh, \myp)} =
  \begin{cases}
    \frac{q}{2} & \text{if } \myp^\prime = \myp, \myh^\prime = \myh - 1 \text{ and } \myh \not= 0\\
    \frac{q}{2} & \text{if } \myp^\prime = \myp, \myh^\prime = \myh \text{ and } \myh = 0\\
    \frac{1-q}{2} & \text{if } \myp^\prime = \myp, \myh^\prime = \myh + 1 \text{ and } \myh \not= l-1\\
    \frac{1-q}{2} & \text{if } \myp^\prime = \myp, \myh^\prime = \myh  \text{ and } \myh = l-1\\
    \frac{r}{2} & \text{if } \myh^\prime = \myh, \myp^\prime = \myp + 1\mod m\\
    \frac{1-r}{2} & \text{if } \myh^\prime = \myh, \myp^\prime = \myp - 1\mod m\\
    0 & \text{otherwise}
  \end{cases}
\end{equation}
where $r$ and $q$ are any two arbitrary real numbers taken in the
interval $(\frac{1}{2}, 1)$. Let us define $\beta = \frac{2q-1}{2}$
the net \textit{vertical} drift felt by the chain.

\begin{rmk}
\label{rmk:sutrdist}
  The transition matrix~\eqref{eq:2s} induces naturally on \Omgn{} a
  graph $G(V, E)$, where $V = \Omgn$ and $\Omgn\times\Omgn \supset E =
  \{(u, v) \::\: \prob{X_n^{t+1} = u \,|\, X_n^{t} = v} > 0\}$. Such
  graph can be thought of as a cylindrical lattice of volume $n$, with $l$ layers composed of $m$ points each.
  Moreover, the neighborhood structure just highlighted introduces a metric on \Omgn{}, given by the length of the shortest path between two vertices of the graph (cfr.\ Remark~\ref{rmk:bad} above).
\end{rmk}

Each chain of the family defined above is an irreducible and aperiodic chain, thus it exists a unique invariant measure \pin{} such that $\pin = \pin \, P_n$.
Since the model has an evident radial symmetry, we expect that
\[\pin(\myh, \myp) = \pin(\myh, \myp^\prime) \quad \forall\: \myp, \myp^\prime \in \{0,1,\ldots,m-1\}\]
Thus let us look for \pin{} in the form
\begin{equation}
  \label{eq:3s}
  \pin(\myh, \myp) = f_n(\myh) \qquad \text{with } f_n(\myh+1) = \alpha f_n(\myh)
\end{equation}
By definition of \pin{} and~\eqref{eq:2s} we have that, for $\myh\not=0,l-1$,
\[\pin(\myh,\myp) = \frac{1-r}{2}\pin(\myh,\myp+1\!\!\!\mod m) + \frac{r}{2}\pin(\myh,\myp-1\!\!\!\mod m) + \frac{q}{2}\pin(\myh+1,\myp) + \frac{1-q}{2}\pin(\myh-1,\myp)\]
which, by virtue of~\eqref{eq:3s}, yields
\begin{equation}
  \label{eq:4s}
  \alpha = 1 \qquad \text{and} \qquad \alpha = \frac{1-q}{q}
\end{equation}
The value of $\alpha$ to be taken is $\alpha = \frac{1-q}{q}$ since it satisfies $\pin = \pin \, P_n$ also for $\myh=0$ and $\myh=l-1$. Thus,
\begin{equation}
  \label{eq:5s}
  \pin(\myh,\myp) = \alpha^{\myh} f_n(0)
\end{equation}
The value of the normalization constant $f_n(0)$ is found via normalization:
\begin{equation}
  \label{eq:6s}
  f_n(0) = \pin(0, \myp) = \frac{2q-1}{m\,q\,(1-\alpha^l)} \simeq \frac{2q-1}{m\,q}
\end{equation}
where last approximation holds for sufficiently large $l$.


Given a state $\Omgn \ni u = (\myh^\prime, \myp^\prime)$, with an abuse of notation
we will denote as $\myh(u)$ and $\myp(u)$ its height, $\myh^\prime$, and its
position on the $\myh^\prime$-th layer, $\myp^\prime$, respectively.

\mbox{}

Consider now the following equivalence relation between any two states $u,v \in \Omgn$
\[u \sim v \quad \Longleftrightarrow \quad h(u) = h(v)\]
The lumped chain, $X_n^{\sharp,t}$, defined on the state space
$\Omgn^\sharp = \{0,1,\ldots,l-1\}$ with transition matrix entries given by
\begin{equation}
  \label{eq:8s}
  P_n^\sharp(i,j) =
  \begin{cases}
    \frac{1}{2} & \text{if } i=j \text{ and } i \not= 0, l-1\\
    \frac{1+q}{2} & \text{if } i=j=0\\
    \frac{2-q}{2} & \text{if } i=j=l-1\\
    \frac{q}{2} & \text{if } j=i-1 \text{ and } i\not=0\\
    \frac{1-q}{2} & \text{if } j=i+1 \text{ and } i\not=l-1\\
    0 & \text{otherwise}
  \end{cases}
\end{equation}
is a projection of $X_n^t$ according to the equivalence relation
$\sim$. The stationary measure $\pin^\sharp(x)$ of the lumped chain is
then found summing $\pin(u)$ over the elements $u$ that belong to the
equivalence class~$[x]$.  
Since every equivalence class (i.e.\ every layer) contains exactly $m$ points:
\begin{equation}
  \label{eq:9s}
  \pin^\sharp(x) \simeq \frac{2q-1}{q}\left(\frac{1-q}{q}\right)^x \qquad x\in\{0,1,\ldots,l-1\}
\end{equation}

\begin{rmk}
The stationary measure $\pin^\sharp$ is obviously reversible with
respect to $P_n^\sharp$ but this property does not hold for the
original chain $X_n^t$, whose equilibrium measure is not reversible
w.r.t.\ $P_n$.
To see this it suffices to take any two states $u, v \in \Omgn$ such that $\myh(u) = \myh(v)$
and $|\myp(u) - \myp(v)| = 1$; then by~\eqref{eq:3s} $\pin(u) =
\pin(v)$ but according to~\eqref{eq:2s} $P(u,v) \not= P(v,u)$.
\end{rmk}

\begin{rmk}
  We have introduced the lumped chain, $X_n^{\sharp,t}$, since it can
  be coupled to $X_n^t$ in such a way that
  \[h(X_n^t) = X_n^{\sharp, t} \qquad \forall\; t\geq0\] Therefore we
  can study the \emph{hitting time of any layer} considering a
  one-dimensional chain only. Nevertheless we want to stress that the
  study of the cutoff phenomenon for $X_n^t$ cannot be reduced to
  the study of the cutoff for $X_n^{\sharp,t}$, since in general the
  identity~\eqref{eq:105} won't hold. 
  Let us consider, indeed, the initial distribution $\mu_n^0 =
  \delta_{u,u_0}$ with $h(u_0) = l-1$, which represents the worst case
  scanario for the behavior of the total variation
  distance. Then~\eqref{eq:105} is false for any finite $t$ but, as we will see, by means
  of Theorem~\ref{thm2} and Corollary~\ref{crl1} it is possible to
  prove cutoff with relative ease.
\end{rmk}

Define now the following family of sets
\begin{equation}
  \label{eq:10s}
  \Ant = \left\{u \in \Omgn \::\: h(u) < \sqrt{\theta}\right\}
\end{equation}
with this definition \Ant{} is the union of the $\sqrt{\theta}$ bottom
layers and $A_{n,1}$ is just the bottommost layer. 
The hitting time \ztheta{\theta} of the set \Ant{} has the following expectation and variance:
\begin{align}
  \label{eq:11s}
  \mean{\ztheta{\theta}} &= \sum_{k=\sqrt{\theta}+1}^l
  \mean{\zeta_{k\to k-1}} = \sum_{k=\sqrt{\theta}+1}^l \frac{2}{q}
  \sum_{i=k}^l \frac{\alpha^i}{\alpha^k}= \beta^{-1}  (l -
  \sqrt{\theta}) + O_\theta(\alpha^l)\\
  \text{Var}\left[\ztheta{\theta}\right] &= \sum_{k=\sqrt{\theta}+1}^l \frac{2}{q}
  \sum_{i=k}^l (2\mean{\zeta_{i\to k-1}}- \mean{\zeta_{k\to
      k-1}})\frac{\alpha^i}{\alpha^k} - \mean{\zeta_{k\to k-1}}= O_\theta(l)
\end{align}
where $\zeta_{i\to j}$ is the first visit time of the state $j$
starting from the state $i$ and $O_\theta( \cdot )$ means $O( \cdot )$
for any fixed value of $\theta$. 

To use Theorem~\ref{thm2} we want to study the behavior of these
quantities in the limit for $n \to \infty$ but $n = l \cdot m$, thus
we can let the volume of the cylinder grow by extending its height or
enlarging its diameter or letting both grow simultaneously. To this
extent let us consider the case where
\begin{equation}
  \label{eq:12s}
  m = n^\omega \quad \text{and}\quad l = n^{1-\omega}
  \qquad\qquad\text{with }\quad \omega>0
\end{equation}
 With the usual notation take $\Delta_n = m^2 = n^{2\omega}$, this
 choice fulfills all the hypothesis of Theorem~\ref{thm2}
 (namely~\eqref{eq:8} and~\eqref{eq:9}) and eventually sets the candidate cutoff-window order to
 \begin{equation}
   \label{eq:13s}
   \delta_n = O(m^2 + \sqrt{l}) = O(n^{2\omega} + n^{\frac{1-\omega}{2}})
\end{equation}


All we are left to deal with is then the existence
(cfr. Corollary~\ref{crl1}) of a coupling $(Z_n^t, W_n^t)$ such that, with $Z_n^0$ located
on a point of the bottommost layer (that is $\myh(Z_n^0) = 0$) and
$W_n^0 \sim \pin$ (i.e.\ $h(W_n^0) \geq 0$ and distributed exponentially), we have
\begin{align}
  \label{eq:14s}
  \lim_{\theta\to\infty}\lim_{n\to\infty}\prob{\gamma_n > \theta\delta_n} = 0
\end{align}
where $\gamma_n = \min\{t\geq0 \::\: Z_n^t = W_n^t\}$ is the coalescence time.

Consider the distance (Cfr. Remark~\ref{rmk:sutrdist}) between $Z_n^t$ and $W_n^t$:
\begin{equation}
  \label{eq:15s}
  D_n^t = |\myh(Z_n^t) - \myh(W_n^t)| + \min\{|\myp(Z_n^t)-\myp(W_n^t)|, m-|\myp(Z_n^t)-\myp(W_n^t)| \}
\end{equation}
It exists a coupling $(Z_n^t, W_n^t)$, sketched for reference Figure~\ref{fig:cpl}, such that
\begin{enumerate}
\item $\myH_n^t = |\myh(Z_n^t) - \myh(W_n^t)|$ is a death-only
  chain on the segment $\{0,1,\ldots, l-1\}$, that is to say
  $H_n^{t+1} \leq H_n^t$
\item $\myH_n^s = 0$ for any $s\geq\gamma_n^\myH = \min\{t\geq0 \::\: \myH_n^t
  = 0\}$
\item\label{pt3} the random time $\gamma_n^\myH$ satisfies $\gamma_n^\myH = \min\{t\geq 0 \::\: h(W_n^t) =
  0\}$
\item\label{pt2} $\myP_n^t = \min\{|\myp(Z_n^t)-\myp(W_n^t)|,
  m-|\myp(Z_n^t)-\myp(W_n^t)| \}$ is a symmetric $r$-lazy random walk
  on the segment $\{0,1,\ldots,\lceil\frac{m}{2}\rceil\}$
\item $\myP_n^s = 0$ for any $s\geq\gamma_n^\myP = \min\{t\geq0 \::\: \myP_n^t
  = 0\}$
\end{enumerate}

\begin{figure}[htbp]
\begin{center}
\includegraphics[width=0.75\textwidth]{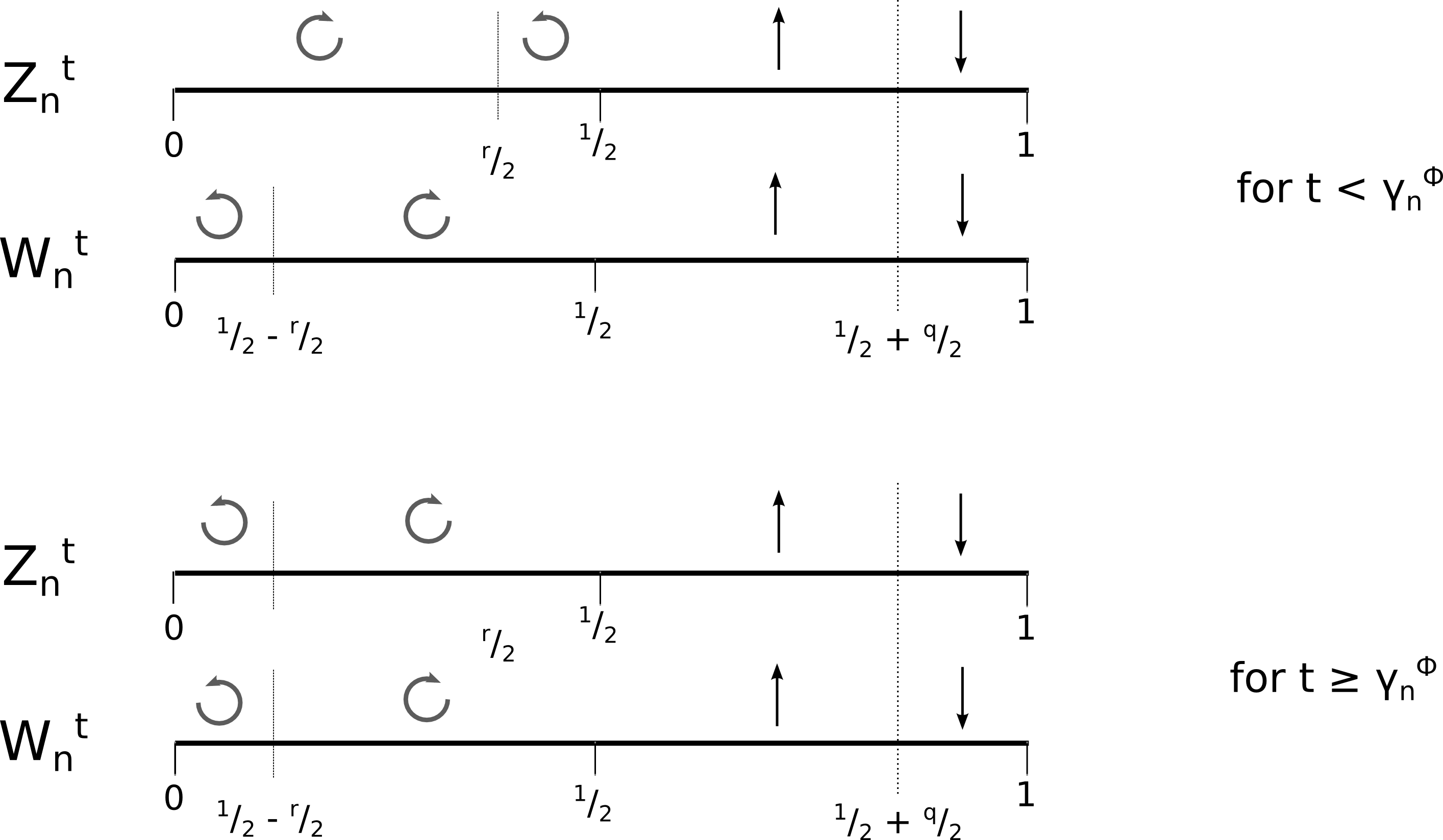}
\caption{Coupling scheme, the same random update is used for both
  $Z_n^t$ and $W_n^t$. The two copies have the same probability
  to move to the upper or lower layer, except when one of the chains
  is on the topmost or bottommost layer. In the latter case the distance
  $\myH_n^t$ has probability $\frac{q}{2}$ to reduce by 1 while in the
  former it has probability $\frac{1-q}{2}$.}
\label{fig:cpl}
\end{center}
\end{figure}
From the description of our coupling it should be clear that
\begin{equation}
  \label{eq:16s}
  \gamma_n = \max\{\gamma_n^\myH, \gamma_n^\myP\} \leq \gamma_n^\myH + \gamma_n^\myP
\end{equation}
Thus, using Markov's inequality we get
\begin{align}
  \prob{\gamma_n \geq \theta\delta_n} &\leq
  \frac{\mean{\gamma_n}}{\theta\delta_n} \leq \frac{\mean{\gamma_n^\myH} +
    \mean{\gamma_n^\myP}}{\theta\delta_n} \label{eq:18s}
\end{align}
Now, according to point~\ref{pt3} listed above and the transition
probabilities of $W_n^{\sharp,t} = h(W_n^t)$ we easily obtain
\begin{equation}
  \label{eq:19s}
  \mean{\gamma_n^\myH \,\left|\, h(W_n^0) = h^\prime\right.} = \beta^{-1} h^\prime
\end{equation}
which yields to
\begin{equation}
  \label{eq:20s}
  \mean{\gamma_n^\myH} = \beta^{-1} \mean{h(W_n^0)} =
  \beta^{-1}\sum_{x} x\,\pin^\sharp(x) \leq \beta^{-1} \frac{1-q}{2q-1}
\end{equation}
According to point~\ref{pt2} listed above we get
\begin{equation}
  \label{eq:21s}
  \mean{\gamma_n^\myP} \propto m^2
\end{equation}
Lines~\eqref{eq:20s} and~\eqref{eq:21s} clearly infer~\eqref{eq:14s} and
the proof is complete: the model exhibit cutoff at cutoff-time 
\begin{equation}
  \label{eq:22s}
  a_n = \beta^{-1} l = \beta^{-1} n^{1-\omega}
\end{equation}
and cutoff-window
\begin{equation}
  \label{eq:23s}
  b_n = O(m^2 + \sqrt{l}) = O\left(n^{2\omega} + n^{\frac{1-\omega}{2}} \right)
\end{equation}

\begin{rmk}
  The condition $\frac{b_n}{a_n} = o(1)$ is fulfilled only if $\omega
  < \frac{1}{3}$. Within this constraint we have cutoff and the
  cutoff-window shows the following behavior:
  \begin{align*}
    0 < \omega \leq \frac{1}{5} &\qquad b_n =
    O\left(n^{\frac{1-\omega}{2}}\right)\\
    \frac{1}{5} \leq \omega < \frac{1}{3} &\qquad b_n =
    O\left(n^{2\omega}\right)\\
  \end{align*}
  and we see that the value $\omega = \frac{1}{5}$ gives the smallest
  cutoff-window order achievable.
\end{rmk}
\begin{rmk}
  The case $\omega=0$ corresponds to an increase of the cylinder volume
  by extending its height while keeping fixed its base diameter, and it
  is almost identical to a biased random walk on a segment~\cite[\S 18.2.1]{levin2006mca}. In
  this sense the general case $\omega > 0$ represents a non-reversible
  higher-dimensional extension of the biased random walk.
\end{rmk}

\subsection{A partially-diffusive random walk}
\label{sec:if-model}

Fix $\varepsilon \in (0,\frac{1}{2})$ and consider the birth-and-death
chain $X_n^t$ defined on the state space $\Omega_n = \{0,1,\ldots,n\}$
with initial position $X_n^0 = n$ and transition rates
\begin{align}
  \label{eq:53a}
  q_i = P_{i,i-1} & =
  \begin{cases}
    \frac{i}{2n} & \quad \text{if } n^\varepsilon < i \leq n\\
    \frac{1}{2}& \quad \text{if } 1 \leq i \leq n^\varepsilon
  \end{cases}\\
  r_i = P_{i,i} &=
  \begin{cases}
    1-p_i-q_i & \quad \text{if } n^\varepsilon\leq i \leq n\\
    0 & \quad \text{if } 0 \leq i < n^\varepsilon
  \end{cases}\\
  p_i = P_{i,i+1} & =
  \begin{cases}
    \frac{i}{4n} & \quad \text{if } n^\varepsilon \leq i \leq n\\
    \frac{1}{2} & \quad \text{if } 0 \leq i < n^\varepsilon
  \end{cases}\label{eq:54a}
\end{align}
This chain is such that outside
the interval $[0,n^\varepsilon]$ it behaves like a biased random walk
while inside the interval it behaves like an
unbiased one. 
It's quite easy to show that this model does not satisfies the
\textit{strong drift condition}, which according
to~\cite{barrera2009abrupt} is a sufficient condition to prove
cutoff, see Remark~\ref{rmk:if} below.
Using Corollary~\ref{crl1} it's easy to show that this model actually
exhibits cutoff. 

The stationary distribution \pin{} can be found by reversibility
\begin{equation}
  \label{eq:180}
  \pin(i)=
  \begin{cases}
    c & \quad\text{for } 0\leq i \leq n^\varepsilon\\
    c \, \left( \frac{1}{2} \right)^{i-n^\varepsilon} & \quad\text{for
    } n^\varepsilon < i \leq n\\
  \end{cases}
\end{equation}
where the constant $c$ is $\frac{1}{n^\varepsilon + 2} + O \left(
  \frac{1}{2 ^{n}} \right) $.
In order to use Theorem~\ref{thm2} it is enough to take the following family of
nested subsets
\begin{equation}
  \label{eq:17a}
  A_{n,\theta} = \{i \::\: 0 \leq i \leq n^\varepsilon
  \theta^{n^{2\varepsilon-1}}\}
\end{equation}
With this choice~\eqref{eq:6} holds and, to the
leading order in $n$
\begin{equation}
  \mean{\ztheta{1}} = \frac{2(1-\varepsilon)}{\log2} \: n\log n \qquad\quad
  \mean{\ztheta{1} -\ztheta{\theta}} = \frac{2}{\log2} n^{2\varepsilon}\log
  \theta \label{eq:60a}\\ 
\end{equation}
see Appendix~\ref{app:mean-if} for the details of the calculations.
Choosing $\Delta_n = n^{2\varepsilon}$ we
verify~\eqref{eq:8} and~\eqref{eq:9}, then by
Remark~\ref{rmk:harm} we know that all the hypotheses hold except possibly~\eqref{eq:7}.

Now we consider a coupling $(Z_n^t, W_n^t)$, where
$Z_n^t$ and $W_n^t$ are two copies of $X_n^t$ with
initial positions $Z_n^0 = n^\varepsilon$ and $W_n^0 \sim \pin$
respectively; then, provided that the two chains have not yet
collided, at each time we let the two copies evolve independently. 
Let $\gamma_n = \min\{t\geq 0 \::\: Z_n^t=W_n^t\}$ be the coalescence
time and set $Z_n^t = W_n^t$ for any $t \geq \gamma_n$, 
then 
\begin{align}
  \prob{\gamma_n > t \,|\, Z_n^0 = n^\varepsilon} 
  &= \prob{\gamma_n > t \,|\, Z_n^0 = n^\varepsilon, W_n^0 \leq
    n^\varepsilon}\prob{W_n^0 \leq n^\varepsilon} + \nonumber\\
  &\quad \prob{\gamma_n > t \,|\, Z_n^0 = n^\varepsilon, W_n^0 >
    n^\varepsilon}\prob{W_n^0 > n^\varepsilon}\label{eq:128}\\
  &\leq \prob{\gamma_n > t \,|\, Z_n^0 = n^\varepsilon, W_n^0 \leq
    n^\varepsilon} + \frac{1}{n^\varepsilon}\label{eq:130}
\end{align}
Let $\tau_n^0 =\min\{t\geq 0 \::\: Z_n^t = 0\}$. Clearly,
\begin{align}
  \prob{\gamma > t \,|\, Z_n^0 = n^\varepsilon, W_n^0 \leq
    n^\varepsilon} &\leq
  \prob{\tau_n^0 > t \,|\, Z_n^t = n^\varepsilon}\label{eq:69a}\\
  &\leq \frac{\mean{\tau_n^0\,|\, Z_n^0 = n^\varepsilon}}{t}\label{eq:72a}
\end{align}
where the last inequality comes from Markov's inequality.
Take  $t=\theta n^{2\varepsilon}$, since $\mean{\tau_n^0\,|\,Z_n^0 = n^\varepsilon} =
n^{2\varepsilon}+O(n^\varepsilon)$, by Remark~\ref{rmk:bad}
\begin{equation}
  \label{eq:199}
  \max_{z_0 \in A_{n,1}}\prob{\gamma_n > \theta n^{2\varepsilon} \,|\, Z_n^0 = z_0} \leq
   \frac{2}{\theta} \text{ definitively as } n\to\infty \tag{\ref{eq:92}a}
\end{equation}

The standard deviation of $\ztheta{1}$ is
$O(n^{1-\frac{\varepsilon}{2}})$ (see Appendix~\ref{app:mean-if}),
therefore \eqref{eq:7}~holds and,
with respect to the coupling defined above,~\eqref{eq:92}
follows from~\eqref{eq:128}-\eqref{eq:72a} with $t=\theta\delta_n =
2\theta\left(n^{2\varepsilon} + n^{1-\frac{\varepsilon}{2}} \right)$.
Thus, by means of Theorem~\ref{thm2} and Corollary~\ref{crl1}
we have that this model exhibits cutoff with cutoff-time 
\begin{equation}
a_n=\mean{\ztheta{1}} = \frac{2(1-\varepsilon)}{\log2} \: n\log n\label{eq:93a}
\end{equation}
and cutoff window
\begin{equation}
  \label{eq:183}
  b_n= 
  \begin{cases}
    O(n^{1-\frac{\varepsilon}{2}}) &\quad \text{if }
  0 < \varepsilon \leq \frac{2}{5}\\
  O(n^{2\varepsilon}) &\quad \text{if }
  \frac{2}{5} < \varepsilon \leq \frac{1}{2}
  \end{cases}
\end{equation}

\begin{rmk}
  From~\ref{eq:183} we see that the choice $\varepsilon = \frac{2}{5}$
  gives the smallest cutoff-window order possible.
\end{rmk}

\begin{rmk}
\label{rmk:crucial}
  This example shows how crucial is the choice of
  $\{A_{n,\theta}\}$. One could try in fact $A_{n,\theta} = \{i \::\: 0 \leq i
  \leq \theta n^\varepsilon\}$, because that scaling, linear on $\theta$,
  worked well in the lazy Ehrenfest chain. This alternative definition would lead
  to an expected travelling time $\mean{\ztheta{1}-\ztheta{\theta}} =
  n\log\theta$ and force $\Delta_n$ (and consequently $\delta_n$) 
  to be of order $n$. Since
  $\theta n$ steps are clearly sufficient for the chain started in
  $n^\varepsilon$ to achieve equilibrium, we would obtain a
  non-optimal $O(n)$ cutoff-window.
\end{rmk}

\begin{rmk}
\label{rmk:if}
  The reason why $X_n^t$ does not satisfies the strong drift
  condition is that it fails the first requirement of the
  definition, namely
  \begin{equation}
    K_q = \inf_{n \in \mathbb{N}} \inf _{0\leq i \leq n} q_i > 0\label{eq:55a}
  \end{equation}
  Nevertheless, it is clear from the results included
  in~\cite{barrera2009abrupt} that the condition $K_q >0$ can actually
  be dropped if one replaces the second condition with
  \begin{equation}
    \label{eq:56a}
    \frac{K^2_n}{K_q^n \, \mean{T^{(n)}_{n \to 0}}} \TendsTo[n, \infty] 0
  \end{equation}
  where $K_q^n = \inf\limits_{0\leq i \leq n} q_i$ and 
  \begin{equation}
    \label{eq:57a}
    K_n = \sup_{1\leq i \leq n} q_i \, \mean{T^{(n)}_{i \to i-1}} =
    \sup_{0\leq i\leq n} \frac{\pin([i,n])}{\pin(i)}
  \end{equation}
  The expected value of $T^{(n)}_{n\to0}$, the hitting time of zero
  starting from $n$, can be easily estimated as
  \begin{equation}
    \mean{T^{(n)}_{n\to0}} \leq O(n\log n + n^{2\varepsilon}) \label{eq:58a}\\
  \end{equation}
  while $K_n$ can be bounded from below by $n^\varepsilon$.  Then
  \begin{equation}
    \label{eq:59a}
    \frac{K^2_n}{K_q^n \, \mean{T^{(n)}_{n \to 0}}} \geq
    \frac{n^{2\varepsilon}}{\frac{n^\varepsilon}{2n}\left[O(n\log n +
        n^{2\varepsilon})\right]} \:\TendsTo[n,\infty] \infty
  \end{equation}
\end{rmk}

\subsection{The mean-field Ising model Glauber dynamics}
\label{sec:mean-field-ising}


The cutoff for the mean-field Ising model evolving according to the
Glauber dynamics has been recently proved
in~\cite{levin2010glauber}. Here we give an alternative proof of the
existence of the cutoff and we evaluate the cutoff-time and the
cutoff-window in terms of an hitting process by means of our
Corollary~\ref{crl1}. The computations needed to achieve this goal in
our framework are quite shortened. A generalization of this result to
the non-symmetrical case, i.e.\ when a constant magnetic field is
added, is likely to be treatable with little effort.

In the mean-field Ising model we have $n$ binary spins and a
neighborhood structure given by a complete graph $K_n$. 
$\mathcal{X}_n = \{+1,-1\}^n$ is the set of all possible
\textit{configurations}. 
The energy of a configuration $\sigma =
(\sigma_1,\sigma_2,\ldots,\sigma_n)$ is then
\begin{equation}
  \label{eq:33}
  H(\sigma) = -\frac{1}{n} \sum_{i<j} \sigma_i\,\sigma_j
\end{equation}
The Glauber dynamics for this model is defined as follows:
\begin{itemize}
\item pick up a site $i\in\{1,2,\ldots,n\}$  uniformly at random
\item update $\sigma_i$ to the values $+1$ or $-1$ respectively with
  probability
  \begin{align}
    p_+ &= \frac{e^{\beta\,S(i)}}{e^{\beta\,S(i)}+e^{-\beta\,S(i)}}\label{eq:34}\\
    p_- &= \frac{e^{-\beta\,S(i)}}{e^{\beta\,S(i)}+e^{-\beta\,S(i)}}\label{eq:107}
  \end{align}
where $S(i) = \frac{1}{n} \sum_{j\not= i} \sigma_j$ is the so-called
\textit{local field}.
\end{itemize}
The parameter $\beta$ has the physical meaning of the inverse
temperature of the system: the higher its value, the stronger
the role of the energy over the entropy in the establishment of the
equilibrium states.
The limiting case of $\beta = 0$
coincides with the lazy random walk on the hypercube seen in
Section~\ref{sec:rwh}: all the spins are updated independently and
they are equivalent from an energy-landscape point of view.

By reversibility it's easy to show that the Markov chain defined above
has a unique stationary measure
\begin{equation}
  \label{eq:109}
  \rho_n(\sigma) = \frac{e^{-\beta\,H(\sigma)}}{Z_{n,\beta}}
\end{equation}
where $Z_{n,\beta} = \sum_{\sigma^\prime\in\Omgn} e^{-\beta\,H(\sigma^\prime)}$ is the
\textit{partition function}.

Let us now define the \textit{magnetization} of a configuration as
\begin{equation}
  \label{eq:110}
  m(\sigma) = \frac{1}{2}\sum_i \sigma_i
\end{equation}
Please note that this is not the standard definition of magnetization,
since the one just defined in~\eqref{eq:110} takes values in 
$\Omgn=\{-\frac{n}{2},-\frac{n}{2}+1,\ldots,
\frac{n}{2}-1,\frac{n}{2}\}$ while in general $m\in[-1,1]$. 
We chose this definition because we want to reduce our system to a
birth-and-death chain.
We can rewrite the Hamiltonian~\eqref{eq:33} in terms of $m(\sigma)$
as follows:
\begin{align}
  m^2(\sigma) &= \frac{1}{4} \left( \sum_i \sigma_i \right) \, \left(
    \sum_j \sigma_j \right)
  = \frac{n}{4} - \frac{n}{2}H(\sigma)\label{eq:117}
\end{align}
and then
\begin{equation}
  \label{eq:118}
  H(m(\sigma)) = -\frac{2m^2(\sigma)}{n} + \frac{1}{2}
\end{equation}
The stationary distribution and the update probabilities take now the
form
\begin{align*}
  &\rho_n(m(\sigma)) = \frac{e^{\frac{2\beta}{n}
      m^2(\sigma)}}{Z^\prime_{n,\beta}} \tag{\ref{eq:109}a}\label{eq:119}\\
 &p_+  = \frac{ e^{\frac{2\beta}{n}( m(\sigma) - \sigma_i )}
      }{e^{\frac{2\beta}{n}( m(\sigma) - \sigma_i
          )}+e^{-\frac{2\beta}{n}( m(\sigma) - \sigma_i )}} 
  = \frac{1}{1+e^{-\frac{4\beta}{n}( m(\sigma) - \sigma_i )}}
  \tag{\ref{eq:34}a}\label{eq:120a}\\ 
&p_-  = \frac{ e^{-\frac{2\beta}{n}( m(\sigma) - \sigma_i )}
      }{e^{\frac{2\beta}{n}( m(\sigma) - \sigma_i
          )}+e^{-\frac{2\beta}{n}( m(\sigma) - \sigma_i )}} 
  = \frac{1}{1+e^{\frac{4\beta}{n}( m(\sigma) - \sigma_i )}}
  \tag{\ref{eq:107}a}\label{eq:121a}
\end{align*}

Let us now define the \textit{magnetization chain}, that is a new
birth-and-death chain $X_n^t$ with state space given
$\Omgn=\{-\frac{n}{2},-\frac{n}{2}+1,\ldots,
\frac{n}{2}-1,\frac{n}{2}\}$ and transition rates
\begin{align}
  p_k &= P_{k,k+1} =
  \frac{\frac{n}{2}-k}{n}\,\frac{1}{1+e^{-\frac{4\beta}{n}( k + 1 )}}\label{eq:120}\\
  q_k &= P_{k,k-1} =
  \frac{\frac{n}{2}+k}{n}\,\frac{1}{1+e^{\frac{4\beta}{n}( k -
      1 )}}\label{eq:123}\\ 
  r_k &= P_{k,k} = \frac{1}{2} + \frac{k}{n} \tanh\left(
    \frac{4\beta}{n}(k+1) \right)\label{eq:124} 
\end{align}
Using standard techniques it is possible to show that the
magnetization chain is actually 
the projection of the Glauber chain according to the following
equivalence relation
\begin{equation}
  \label{eq:132}
  \sigma \sim \sigma^{\prime} \quad\Longleftrightarrow\quad
  m(\sigma)=m(\sigma^{\prime})
\end{equation}
see for example~\cite[Thm.~5.1.4.1]{presutti2008scaling}.

Consider the Glauber chain, started say with initial distribution
$\lambda_n^0$ on $\mathcal{X}_n$ such that $\lambda_n^0(\sigma)
= \lambda_n^0(\sigma^\prime)$ whenever $\sigma\sim\sigma^\prime$. 
Along with this process take its projection, the magnetization
chain, that has initial distribution $\mu_n^0$ and stationary measure
\pin{} equal to
\begin{align}
  \label{eq:133}
  \mu_n^0(k) &= \sum_{\sigma\::\:m(\sigma)=k}\lambda_n^0(\sigma)\\
  \pi_n(k) &= \sum_{\sigma\::\:m(\sigma)=k}\rho_n(\sigma) = \frac{e^{
      \frac{2\beta k^2}{n} }}{Z_{n,\beta}} {n \choose 
  \frac{n}{2}+k}\label{eq:131}
\end{align}
It is not difficult whatsoever to prove that $\lambda_n^0(\sigma) =
\lambda_n^0(\sigma^\prime)$ for $\sigma\sim\sigma^\prime$ leads to 
$\lambda_n^t(\sigma) = \lambda_n^t(\sigma^\prime)$ for any $t\geq 0$,
which in turn infers that
\begin{equation}
  \label{eq:135}
  \dist{\lambda_n^t}{\rho_n} = \dist{\mu_n^t}{\pin} \quad \forall\,t\geq0
\end{equation}
In other words, the Glauber chain exhibit cutoff if and only if the
magnetization chain does.

\subsubsection{Analysis of $\pi_n(k)$}
\label{sec:analysis-pin}

Fix $\theta\geq1$ and define 
\begin{equation}
\label{eq:136}
\Ant = \left\{
k \in \Omgn \::\: -\theta\sqrt{\frac{n}{1-\beta}} \leq k 
\leq \theta\sqrt{\frac{n}{1-\beta}} \right\}
\end{equation}
For $k \in \Ant$ we can estimate $\pin(k)$ by means of
Stirling's formula:
\begin{align}
  {n \choose \frac{n}{2}+k} 
  &=\frac{2^{n+\frac{1}{2}}}{\sqrt{\pi n\left(1-\frac{4k^2}{n^2}\right)}}
  \frac{(1+O\left(n^{-1}\right))}{\left(1+\frac{2k}{n}\right)^{\frac{n}{2}\left(1+\frac{2k}{n}\right)}
  \left(1-\frac{2k}{n}\right)^{\frac{n}{2}\left(1-\frac{2k}{n}\right)}}
\label{eq:138}
\end{align}
Next we pass to the \emph{log} and use its analytic expansion to get
\begin{align}
  \log \frac{1}{\left(1+\frac{2k}{n}\right)^{\frac{n}{2}\left(1+\frac{2k}{n}\right)}
  \left(1-\frac{2k}{n}\right)^{\frac{n}{2}\left(1-\frac{2k}{n}\right)}}
  =-\frac{n}{2}\left[\sum_{i\geq1}\left(\frac{2k}{n}\right)^{2i}
\frac{1}{2i^2-i}\right]\label{eq:140}
\end{align}
Therefore for $k \in \Ant$ we have
\begin{align}
  \pi_n(k) 
  &=\frac{2^{n+\frac{1}{2}}}{Z_{n,\beta}} \sqrt{ \frac{1}{\pi
    n } } e^{ -\frac{2(1-\beta)}{n}
  k^2} (1+O\left(n^{-1}\right))\label{eq:142}
\end{align}
that is $\pi_n(k)$ is very close to a Gaussian distribution
$\mathcal{N}\left(0,\frac{1}{2}\sqrt{\frac{n}{1-\beta}}\right)$ for 
$k\in\Ant$. This means that~\eqref{eq:6} holds, because there exists a
positive constant $c_\beta$ such that, for $n$ sufficiently large
\begin{equation}
  \pin\left(\Ant^\complement\right) < \frac{c_\beta}{\theta^2}\label{eq:143}
\end{equation}
\begin{rmk}
  Note that in this model the Gaussian structure of $\pin$ is given by
  both energy and entropy contribution, merging in the expression of
  the free-energy, which can be recognized as the exponent of
  $e^{-\frac{2(1-\beta)}{n} k^2}$ divided by $\beta$. Hence in this
  case we will say that the cutoff is \textit{free-energy driven}.
\end{rmk}

\subsubsection{Proof of cutoff}
\label{sec:proof-cutoff}

Now suppose the Glauber chain
is started at time $0$ with magnetization $\frac{n}{2}$, that is
$\lambda_n^0 = \delta_{\sigma, 1}$ and $\mut{0} = \delta_{i,\frac{n}{2}}$;
this choice gives equal probability to equivalent configurations, 
then~\eqref{eq:135} holds. 
As usual define $\ztheta{\theta}$ as the
hitting time of \Ant{} and \ztheta{1}{} as the hitting time of $A_{n,1}$.
Lengthy but straightforward calculations (deferred to
Appendix~\ref{app:mean-var}) show that, to the leading order in $n$
\begin{align}
  &\mean{\ztheta{1}} = \frac{1}{2(1-\beta)}n\log n \label{eq:144}\\
  &\mean{\ztheta{1}-\ztheta{\theta}} = (1+\log\theta) \, O(n)\label{eq:145}
\end{align}
and that $\text{Var}[\ztheta{1}]$ grows at most as
$O(n^2)$; Therefore hypotheses~\eqref{eq:7}-\eqref{eq:16} are satisfied.
Moreover, choose $\Delta_n = O(n)$, $\delta_n$ is now of order $n$ and
both~\eqref{eq:8} and~\eqref{eq:9} are fulfilled, so that
Theorem~\ref{thm2} gives us~\eqref{eq:11}. 
Then we are left to prove that with $\delta_n = O(n)$
Corollary~\ref{crl1} holds.
\begin{rmk}
  \label{rmk:beta}
  Since for $\beta = 0$ the magnetization chain reduces to the
  Ehrenfest chain, the following estimates hold as well for the
  Ehrenfest Urn model presented in Section~\ref{sec:ehrenfest}.
\end{rmk}

To prove Corollary~\ref{crl1} consider the following coupling, $(Z_n^t, W_n^t, Z_n^{+,t},
Z_n^{-,t})$ where each component is a copy of the magnetization chain and
\begin{equation}
  \label{eq:122}
  \begin{array}{lcl}
    Z_n^0 = z_0 \equiv \frac{1}{2}\sqrt{\frac{n}{1-\beta}} & \quad
    &Z_n^{+,t} = z_0^+ \equiv
    z_0 \theta^{\frac{1}{3}}\\
    W_n^0 \sim \pin &\quad& Z_n^{-,t} = z_0^- \equiv
    -z_0 \theta^{\frac{1}{3}}
  \end{array}
\end{equation}
for a given fixed $\theta > 1$.
Let any of the four chains move according with the same transition
probabilities and using the same 
i.i.d.\ random update $u\sim U(0,1)$. 
To illustrate the transition probabilities let us consider for
instance the chain $Z_n^t$ and
suppose that at time $t$ we have $Z_n^t=k$, then
\begin{align*}
  \text{if } k\geq0 & \qquad
  \begin{cases}
    Z_n^{t+1} = Z_n^t + 1 & \text{if } 0\leq u < p_k\\
    Z_n^{t+1} = Z_n^t & \text{if } p_k \leq u \leq 1-q_k\\
    Z_n^{t+1} = Z_n^t - 1 & \text{if } 1-q_k < u \leq 1
  \end{cases}\\
  \text{if } k<0 & \qquad
  \begin{cases}
     Z_n^{t+1} = Z_n^t - 1 & \text{if } 0\leq u < q_k\\
    Z_n^{t+1} = Z_n^t & \text{if } q_k \leq u \leq 1-p_k\\
    Z_n^{t+1} = Z_n^t + 1 & \text{if } 1-p_k < u \leq 1
  \end{cases}
\end{align*}
The restriction of the coupling defined above to its first two
components, $Z_n^t$ and $Y_n^t$, is the coupling we are going to consider for
Corollary~\ref{crl1}. Thus we define $\gamma_n=\min\{t\geq0
\,:\, Z_n^t = W_n^t\}$ and recall Remark~\ref{rmk:bad}.

  By a careful analysis of~\eqref{eq:120}-\eqref{eq:124} (noticing, in particular,
  that $r_k \geq \frac{1}{2}$ and $p_k = q_{-k}$) such a scheme ensures
  that any two components of the coupling mantain their relative
  partial order undergoing a single-step transition, and indeed it is impossible that
  two chains at distance $1$ will undergo a one step transition that
  would change their relative order.

Hence the evolution scheme described
above has the following sandwiching properties 
\begin{enumerate}
\item $Z_n^{+,t} = -Z_n^{-,t}$
\item $Z_n^{-,t} \leq Z_n^t \leq Z_n^{+,t}$
\item $Z_n^{-,t} \leq W_n^t \leq Z_n^{+,t}$ provided that $W_n^0 \in A_{n,\theta^{\frac{1}{3}}}$
\end{enumerate}

Using~\eqref{eq:143} we have
\begin{align}
   \prob{\gamma_n > t \,|\, Z_n^0 = z_0}
   &=\prob{ \gamma_n > t \,\left|\,  Z_n^0 = z_0, W_n^0 \in
       A_{n,\theta^{\frac{1}{3}}} \right. }\prob{W_n^0 \in
       A_{n,\theta^{\frac{1}{3}}}} + \nonumber\\
   &\prob{ \gamma_n > t \,\left|\,  Z_n^0 = z_0, W_n^0 \not\in
       A_{n,\theta^{\frac{1}{3}}} \right.}\prob{W_n^0 \not\in
       A_{n,\theta^{\frac{1}{3}}}} \label{eq:131}\\
   &\leq \prob{ \gamma_n > t \,\left|\,  Z_n^0 = z_0, W_n^0 \in
       A_{n,\theta^{\frac{1}{3}}} \right.} +
     \frac{c_\beta}{\theta^{\frac{2}{3}}} \label{eq:132}
\end{align}
Therefore by means of the sandwiching properties stated above
\begin{align}
  \prob{ \gamma_n > t \,\left|\,  Z_n^0 = z_0, W_n^0 \in
       A_{n,\theta^{\frac{1}{3}}} \right.} \leq \prob{\tau^0_n > t \,|\,
       Z_n^{+,0} = z_0^+} \label{eq:175}
\end{align}
where $\tau_n^0 = \min\{t\geq0 \::\: Z_n^{+,t} = Z_n^{-,t} = 0\}$.
Note that $Z_n^{+,t}$ has a drift towards 0 as well as any copy of the
magnetization chain. Accordingly, it can be coupled with a lazy
uniform random walk $R_n^t$ such that
\begin{align}
  &R_n^0 = z_0^+\label{eq:176}\\
  &\prob{\tau_n^0 > t \,|\, Z_n^{+,0} = z_0^+} \leq \prob{\tilde{\tau}_n^0  > t
    \,|\, R_n^{0} = z_0^+}\label{eq:177}
\end{align}
where $\tilde{\tau}_n^0 = \min\{t\geq0 \::\: R_n^t = 0\}.$ Now we can use the following
estimate, which is a classical result for random walks
\begin{align}
  \prob{\tilde{\tau}_n^0  > t
    \,|\, R_n^{0} = z_0^+} \leq \frac{c\,z_0^+}{\sqrt{t}}\label{eq:178}
\end{align}
and we have found that Corollary~\ref{crl1} holds with $\delta_n = n$.

\begin{appendices}
\section{Mean value and variance of \ztheta{1}{} for the mean-field Ising model}
\label{app:mean-var}
In this appendix we present in full details the estimates for $\mean{\ztheta{1}}$ and
$\text{Var}[\ztheta{1}]$ we have used to apply Corollary~\ref{crl1} to the
magnetization chain in Section~\ref{sec:mean-field-ising}. Since
for $\beta = 0$ the magnetization chain reduces to the Ehrenfest
chain, the following estimates hold as well for the Ehrenfest Urn
model presented in Section~\ref{sec:ehrenfest}.

Standard formulas (see e.g.~\cite{barrera2009abrupt}) give
\begin{align}
  \mean{\ztheta{1}} &= \sum_{k=\kinf+1}^\ksup \mean{\zeta_{k\to
    k-1}} =\sum_{k=\kinf+1}^\ksup \frac{1}{q_k} \sum_{j=k}^\ksup \frac{\pin(j)}{\pin(k)}
  \label{eq:147}\\
  \text{Var}[\ztheta{1}] &= \sum_{k=\kinf+1}^\ksup \text{Var}[\zeta_{k\to
    k-1}]\nonumber\\ 
  &= \sum_{k=\kinf+1}^\ksup \frac{1}{q_k}
  \sum_{j=k}^\ksup \left(2\mean{\zeta_{j\to k-1}} -\mean{\zeta_{k\to
        k-1}} \right)\frac{\pin(j)}{\pin(k)} - \mean{\ztheta{1}} \label{eq:121}
\end{align}
where $\zeta_{k\to k-1}$ is the first time the chain visits $k-1$ after
visiting $k$ and
\begin{align}
  \label{eq:208}
  &q_k = \frac{n}{\nhpk}\left(1+e^{\frac{4\beta}{n}(k-1) }\right)\\
  \label{eq:210}
  &\frac{\pin(j)}{\pin(k)} = \frac{ {n \choose \frac{n}{2}+j} }{ {n \choose \nhpk} }
  e^{\frac{2\beta}{n} (j^2-k^2)}
\end{align}

Let us begin rewriting the ratio of the two binomial coefficients as
\begin{align}
  \frac{ {n \choose \frac{n}{2}+j} }{ {n \choose \nhpk} } &= 
  \prod_{i=0}^{j-k-1}
  \frac{\frac{n}{2}-k-i}{\frac{n}{2}+k+i+1}\nonumber\\
  &= \left(\frac{\frac{n}{2}-k}{\nhpk+1}\right)^{j-k}\;
  \prod_{i=0}^{j-k-1}\left(1-\frac{i}{\frac{n}{2}-k}\right)\left(\frac{1}{1+\frac{i}{\frac{n}{2}+k+1}}\right)\label{eq:134}
\end{align}
Next, note that for any of the values of triple $(i,j,k)$ involved in the calculations
\begin{equation}
  \label{eq:146}
  0 \leq \frac{i}{\frac{n}{2}+k+1} \leq \frac{i}{\frac{n}{2}-k} \leq 1
\end{equation}
So we find handy the following two easy lemmas.
\begin{lmm}
\label{lmm:exp1}
  For $x \in [0,1]$
  \begin{equation}
    \label{eq:148}
    (1-x)\frac{1}{1+x} \leq e^{-2x}
  \end{equation}
\end{lmm}
\begin{lmm}
\label{lmm:exp2}
  For $0\leq y \leq x \leq 1$
  \begin{equation}
    \label{eq:152}
    (1-x)\frac{1}{1+y} \leq e^{-x-y}
  \end{equation}
\end{lmm}

In virtue of Lemma~\ref{lmm:exp2} we can bound line~\eqref{eq:134} as
follows:
\begin{align}
  \frac{ {n \choose \frac{n}{2}+j} }{ {n \choose \nhpk} }&\leq 
  \left(\frac{\frac{n}{2}-k}{\frac{n}{2}+k}\right)^{j-k}\,
  e^{-\sum_{i=0}^{j-k-1} \left[\frac{i}{\frac{n}{2}-k} +
      \frac{i}{\frac{n}{2}+k+1}\right]}  \label{eq:155}\\
  &= \left(\frac{\frac{n}{2}-k}{\frac{n}{2}+k}\right)^{j-k}\,
  e^{ \frac{ -2(j-k)^2+2(j-k) }{ n\left(
        1-\frac{4k^2}{n^2}+\frac{2}{n}-\frac{4k}{n^2} \right) } }
  e^{ \frac{ -2(j-k)^2+2(j-k) }{ n^2\left(
        1-\frac{4k^2}{n^2}+\frac{2}{n}-\frac{4k}{n^2} \right) }
  }\label{eq:156}\\
  &\leq \left(\frac{\frac{n}{2}-k}{\frac{n}{2}+k}\right)^{j-k}\,
  e^{ \frac{ -2(j-k)^2+2(j-k) }{ n\left(
        1-\frac{4k^2}{n^2}+\frac{2}{n}-\frac{4k}{n^2} \right) } }\label{eq:157}
\end{align}
Thus, for $\kinf \leq k \leq \frac{n}{2}-\log n$,
\begin{align}
  \sum_{j=k}^\ksup \frac{ {n \choose \frac{n}{2}+j} }{ {n \choose
      \nhpk} } e^{\frac{2\beta}{n}(j^2-k^2)} &\leq
  \sum_{j=k}^\ksup \left(\frac{\frac{n}{2}-k}{\frac{n}{2}+k}\right)^{j-k}\,
  e^{ \frac{ -2(j-k)^2+2(j-k) }{ n\left(
        1-\frac{4k^2}{n^2}+\frac{2}{n}-\frac{4k}{n^2} \right) } }
  e^{\frac{2\beta}{n}(j^2-k^2)}\label{eq:158}\\
  &\leq \sum_{l=0}^\ksupk
  \left(\frac{\frac{n}{2}-k}{\frac{n}{2}+k}\,e^{\frac{4\beta
        k}{n}}\right)^l\,
  e^{\frac{-2(1-\beta)l^2 + 2l }{ n\left( 1-\frac{4k^2}{n^2} +
        \frac{2}{n} - \frac{4k}{n^2} \right) }}\label{eq:160}\\
  &\leq \sum_{l=0}^\infty
  \left(\frac{\frac{n}{2}-k}{\frac{n}{2}+k}\,e^{\frac{4\beta
        k}{n}}\right)^l \,\left(
    1+O\left(\log^{-1}n\right)\right) \label{eq:161}\\
  &=\frac{\left(\nhpk \right) \left(
      1+O\left(\log^{-1}n\right)\right)}{\frac{n}{2}\left(1-e^{\frac{4\beta
          k}{n}}\right)+ 
    k\left(1+e^{\frac{4\beta k}{n}}\right)} \label{eq:162}\\
  &\leq \frac{\nhpk}{2(1-\beta)k} \, \left(
    1+O\left(\log^{-1}n\right)\right)\label{eq:163}
\end{align}
Therefore we obtain the following upper bounds:
\begin{align}
  \mean{\ztheta{1}} &\leq \frac{n}{2(1-\beta)}
  \left[\sum_{k=\kinf+1}^{\ksup-\log n}
  \frac{\fone}{k}\right] \, \left(
  1+O\left(\log^{-1}n\right)\right) \nonumber\\
  &\qquad + c_1\,\sum_{k=\frac{n}{2}-\log n}^\ksup
  \frac{n}{\nhpk}\left(1+e^{\frac{4\beta}{n}(k-1) }\right)
  \sum_{l=0}^{\ksup-k} \left( c_2\frac{\log n}{n} \right)^l\label{eq:164}\\
  &= \frac{1}{2(1-\beta)} n \log n + O(n)\label{eq:165} 
\end{align}
and
\begin{align}
  \label{eq:189}
  \mean{\ztheta{1}-\ztheta{\theta}} &\leq \frac{n}{2(1-\beta)}
  \left[\sum_{k=\kinf+1}^{\frac{\theta}{2}\sqrt{\frac{n}{1-\beta}}} \frac{\fone}{k}\right] + O(n)\\
  &=  (1 + \log \theta) O(n)\label{eq:190}
\end{align}

From previous computations, noticing that
\begin{equation}
  e^{\frac{-2(1-\beta)l^2 +
    2l }{ n\left( 1-\frac{4k^2}{n^2} + \frac{2}{n} - \frac{4k}{n^2}
    \right) }} \leq \sqrt{e} \label{eq:191}
\end{equation}
we have that
\begin{equation}
  \label{eq:167}
  \mean{\zeta_{k\to k-1}} \leq \sqrt{e} \: \frac{\fone}{2(1-\beta)}\frac{n}{k}
\end{equation}
and then by summation
\begin{equation}
  \label{eq:168}
  \mean{\zeta_{k+l\to k-1}} \leq \frac{\sqrt{e}}{2(1-\beta} n \log\left(1+\frac{l}{k}\right) + O(n)
\end{equation}
From~\eqref{eq:121}, using~\eqref{eq:157} and~\eqref{eq:167}-\eqref{eq:168}, 
we can easily bound the variance of $\zeta_{k\to k-1}$ as follows:
\begin{align}
  \text{Var}[\zeta_{k\to k-1}] &\leq \frac{n}{\nhpk}\left(1+e^{\frac{4\beta
        }{n}(k-1)}\right) \sum_{j=k}^\ksup \mean{\zeta_{j\to
      k-1}} \frac{ {n\choose \frac{n}{2}+j} }{ {n \choose \nhpk} }
  e^{\frac{2\beta}{n} (j^2-k^2)}\label{eq:166}\\
  & \leq \frac{c_\beta \,  n^2}{\nhpk} \sum_{l=0}^\ksupk
  \log\left(1+\frac{l}{k}\right)
  \left(\frac{\frac{n}{2}-k}{\frac{n}{2}+k} \, e^{\frac{4\beta
        k}{n}}\right)^l \label{eq:169}\\
  &\leq \frac{c_\beta \,  n^2}{\nhpk} \sum_{l=0}^\ksupk \frac{l}{k}
  \left(\frac{\frac{n}{2}-k}{\frac{n}{2}+k} \, e^{\frac{4\beta k}{n}}
  \right)^l \label{eq:170}\\
  &\leq \frac{c_\beta \,  n^2}{\nhpk} \, \frac{1}{k} \, \sum_{l=0}^\infty l
  \left(\frac{\frac{n}{2}-k}{\frac{n}{2}+k} \, e^{\frac{4\beta k}{n}}
  \right)^l \label{eq:171}\\
  & \leq \frac{c_\beta \,  n^2}{\left(\nhpk\right)^2} \, \frac{1}{k} \,
 \left( \ksupk\right) \,
  \left(\frac{\nhpk}{2(1-\beta)k}\right)^2\label{eq:172}\\
  & \leq c_\beta \frac{n^3}{k^3}\label{eq:173}
\end{align}
Therefore $\text{Var}[\ztheta{1}] = \sum_{k=\kinf+1}^\ksup
\text{Var}[\zeta_{k\to k-1}] $ grows at most as $O(n^2)$. 

Eventually, let us bound from below the expectation
$\mean{\ztheta{1}}$.
From~\eqref{eq:147} and~\eqref{eq:134} we have
\begin{align}
  \mean{\ztheta{1}} \geq \sum_{\kinflog}^{\ksuplog}
  \frac{n}{\nhpk}\left( 1+e^{\frac{4\beta k}{n}(k-1)} \right) 
  \sum_{j=k}^{k+\ksuploglog} \Bigg[ \left(
    \frac{\frac{n}{2}-k}{\frac{n}{2}+k} \right)^{j-k} \nonumber\\
   \prod_{i=0}^{j-k-1}\left(1-\frac{i}{\frac{n}{2}-k}\right)
\left(\frac{1}{1+\frac{i}{\frac{n}{2}+k+1}}\right) 
e^{\frac{2\beta}{n} (j^2-k^2)}\Bigg] \label{eq:71}
\end{align}
Then we have
\begin{align}
  &\sum_{j=k}^{k+\ksuploglog}  \left(
    \frac{\frac{n}{2}-k}{\frac{n}{2}+k} \right)^{j-k}
   \prod_{i=0}^{j-k-1}\left(1-\frac{i}{\frac{n}{2}-k}\right)
\left(\frac{1}{1+\frac{i}{\frac{n}{2}+k+1}}\right) 
e^{\frac{2\beta}{n} (j^2-k^2)} \nonumber\\
  \label{eq:112}
&\geq\sum_{j=k}^{k+\ksuploglog} \left(
    \frac{\frac{n}{2}-k}{\frac{n}{2}+k} \right)^{j-k} 
   \left(\prod_{i=0}^{j-k-1} \left(1-\frac{i}{\frac{n}{2}-k}\right)
     e^{-\frac{i}{\frac{n}{2}+k+1}}\right) 
   e^{\frac{2\beta}{n} (j^2-k^2)}\\
  \label{eq:77}
  &\geq\sum_{j=k}^{k+\ksuploglog} \left[\left(
    \frac{\frac{n}{2}-k}{\frac{n}{2}+k} \right)^{j-k}
   \left(\prod_{i=0}^{j-k-1} e^{-\frac{i}{\frac{n}{2}-k} -
       \frac{i}{\frac{n}{2}+k+1}}\right) 
   e^{\frac{2\beta}{n} (j^2-k^2)} + \varepsilon_1 \right]
\end{align}
where $\varepsilon_1$ tends to $0$ exponentially fast in $n$.
\begin{rmk}
  The error $\varepsilon_1$ gives a negligible contribution
  to $\mean{\ztheta{1}}$ being exponentially small, for this reason
  we will henceforth drop it. 
\end{rmk}
The right-hand in~\eqref{eq:77} can be rewritten as follows
\begin{align}
  \label{eq:93}
  &\sum_{j=k}^{k+\ksuploglog} \left(
    \frac{\frac{n}{2}-k}{\frac{n}{2}+k} \right)^{j-k}
   e^{ \frac{ -2(j-k)^2+2(j-k) }{ n\left(
        1-\frac{4k^2}{n^2}+\frac{2}{n}-\frac{4k}{n^2} \right) } +
    \frac{ -2(j-k)^2+2(j-k) }{ n^2\left(
        1-\frac{4k^2}{n^2}+\frac{2}{n}-\frac{4k}{n^2} \right) }
    +\frac{2\beta}{n} (j^2-k^2)} \\
  \label{eq:94}
  &= \sum_{l=0}^{\ksuploglog} \left(
    \frac{\frac{n}{2}-k}{\frac{n}{2}+k} e^{\frac{4\beta k}{n}}\right)^{l}
   e^{ \frac{ -2l^2+2l }{ n\left(
        1-\frac{4k^2}{n^2}+\frac{2}{n}-\frac{4k}{n^2} \right) } +
    \frac{ -2l^2+2l }{ n^2\left(
        1-\frac{4k^2}{n^2}+\frac{2}{n}-\frac{4k}{n^2} \right) }
    +\frac{2\beta}{n} l^2}\\
  \label{eq:95}
  &= \sum_{l=0}^{\ksuploglog} \left(
    \frac{\frac{n}{2}-k}{\frac{n}{2}+k} e^{\frac{4\beta k}{n}}\right)^{l}
   e^{ \frac{ -2l^2+2l }{ n\left(
        1-\frac{4k^2}{n^2}+\frac{2}{n}-\frac{4k}{n^2} \right) } 
    +\frac{2\beta}{n} l^2} (1 + \varepsilon_2)
\end{align}
with $\varepsilon_2 = o\left( n^{-1} \right)$. Now set $\varphi =
-\frac{4k^2}{n^2}+\frac{2}{n}-\frac{4k}{n^2}$, then~\eqref{eq:95} can
be rewritten as follows
\begin{align}
  \label{eq:96}
  &\sum_{l=0}^{\ksuploglog} \left(
    \frac{\frac{n}{2}-k}{\frac{n}{2}+k} e^{\frac{4\beta k}{n}}\right)^{l}
   e^{ \frac{ -2(1-\beta(1+\varphi))l^2}{n(1+\varphi)} 
    +\frac{2l}{n(1+\varphi)} } (1+ \varepsilon_2)\\
  \label{eq:108}
  &= \sum_{l=0}^{\ksuploglog} \left(
    \frac{\frac{n}{2}-k}{\frac{n}{2}+k} e^{\frac{4\beta
        k}{n}}\right)^{l} (1 -
  \varepsilon)(1+\varepsilon_3)(1+\varepsilon_2)\\
  \label{eq:114}
 &= \sum_{l=0}^{\ksuploglog} \left(
    \frac{\frac{n}{2}-k}{\frac{n}{2}+k} e^{\frac{4\beta
        k}{n}}\right)^{l} (1 - \varepsilon)
\end{align}
where $\varepsilon = O\left(\log^{-2}(\log n) \right)$ and
$\varepsilon_3 = O\left( n^{-\frac{1}{2}}\log\log n \right)$.

Therefore
\begin{align}
  &\mean{\ztheta{1}} \geq (1-\varepsilon)\sum_{\kinflog}^{\ksuplog}
  \frac{n}{\nhpk}\left( 1+e^{\frac{4\beta k}{n}(k-1)} \right)
  \sum_{l=0}^{\ksuploglog} \left( \frac{\frac{n}{2}-k}{\frac{n}{2}+k}
    e^{\frac{4\beta k}{n}}\right)^{l} \nonumber\\
  &= \mean{\ztheta{1}} \geq (1-\varepsilon)\sum_{\kinflog}^{\ksuplog}
  \Bigg\{ \frac{n}{\nhpk} \left( 1+e^{\frac{4\beta k}{n}(k-1)} \right)
  \nonumber\\
  &\qquad\frac{(\nhpk+1)}{\frac{n}{2}(1-e^{\frac{4\beta k}{n}}) +
    k(1+e^{\frac{4\beta k}{n}})+1} \left[ 1-\left(
      \frac{\frac{n}{2}-k}{\nhpk+1} e^{\frac{4\beta k}{n}}
    \right)^{1+\frac{\sqrt{n}}{\log\log n}}
  \right] \Bigg\} \nonumber\\
  &\geq (1-\gamma) (1-\varepsilon)\sum_{\kinflog}^{\ksuplog} \frac{n}{\nhpk} \left(
    2+O\left( \log^{-1}n \right) \right) 
  \,\Gamma
  \label{eq:125}
\end{align}
where
\begin{equation}
  \label{eq:184}
    \gamma = \left[\frac{1-\frac{\log
          n}{\sqrt{n(1-\beta)}}}{1+\frac{\log
          n}{\sqrt{n(1-\beta)}}} \left( 1+2\beta \frac{\log
          n}{\sqrt{n(1-\beta)}} + O\left( \frac{\log^2n}{n} \right) \right)
    \right]^{1+\frac{\sqrt{n}}{\log\log n}}
\end{equation}
and
\begin{equation}
  \label{eq:185}
  \Gamma = \frac{\nhpk}{\frac{n}{2}(\frac{-4\beta k}{n}+O\left(
      \log^{-2}n \right)) + 
    k(2+\frac{4\beta k}{n} + O\left( \log^{-2}n \right))+2}
\end{equation}

Now, $\gamma$ can be rewritten as
\begin{equation}
  \label{eq:186}
  \gamma = \left[ 1-\frac{2(1-\beta)\log
      n}{\sqrt{n(1-\beta)}} +  O\left(
      \frac{\log^2n}{n}\right)\right]^{1+\frac{\sqrt{n}}{\log\log n}}  
\end{equation}
Therefore $\gamma$ tends asintotically to $0$.

The right-hand in~\eqref{eq:125} now becomes
\begin{equation}
  \label{eq:187}
  (1-\gamma) (1-\varepsilon)\sum_{\kinflog}^{\ksuplog} 
\frac{2n \left( 1+O\left( \log^{-1}n \right) \right) }{2k(1-\beta) + 2
  + O\left( \log^{-1}n \right)} 
\end{equation}
from which we see that, to the leading order in $n$
\begin{equation}
  \label{eq:188}
  \mean{\ztheta{1}} \geq \frac{1}{2(1-\beta)}n\log n 
\end{equation}

\section{Mean value and variance of \ztheta{1}{} for the partially
  diffusive random walk}
\label{app:mean-if}

Standard formulas (see e.g.~\cite{barrera2009abrupt}) give
\begin{align}
  \label{eq:192}
  \mean{\ztheta{1}} &= \sum_{k=n^\varepsilon+1}^n \mean{\zeta_{k\to k-1}} =\sum_{k=n^\varepsilon+1}^n \frac{2n}{k} \sum_{m=k}^n \frac{\pin(m)}{\pin(k)}
\end{align}
where $\zeta_{k\to k-1}$ is the first time the chain visits $k-1$ after
visiting $k$. By means of~\eqref{eq:53a}-\eqref{eq:54a} and reversibility,
\begin{align}
  \label{eq:76}
  \phi(k) &= \sum_{m=k}^n \frac{\pin(m)}{\pin(k)}\\
  \label{eq:193}
  &= \sum_{m=k}^n \frac{k}{m} 2^{k-m}\\
  \label{eq:194}
  &\simeq k2^k \int_{-k\log2}^{-n\log2} \frac{e^t}{t} \, dt
\end{align}
Using the properties of the exponential integral we get
\begin{equation}
  \label{eq:195}
  \phi(k) = \frac{1}{\log2} - \frac{k}{n\log2}2^{(k-n)} + O\left(\frac{1}{k}\right)
\end{equation}
and therefore
\begin{equation}
  \label{eq:196}
  \mean{\ztheta{1}} = \frac{2(1-\varepsilon)}{\log2} n \log n +
  O\left( n^{1-\varepsilon}\right)
\end{equation}
Similarly, for $n$ sufficiently large we have that
\begin{equation}
  \label{eq:198}
  \mean{\ztheta{1}-\ztheta{\theta}} = \sum_{k=n^\varepsilon +
    1}^{n^\varepsilon \theta^{n^{2\varepsilon-1}}} \frac{2n}{k}
  \phi(k) = \frac{2 n^{2\varepsilon}}{\log2} \log\theta + O\left(
    n^\varepsilon\log\theta \right)
\end{equation}
From~\eqref{eq:198} we see that for $n$ sufficiently large
$\mean{\ztheta{1}-\ztheta{\theta}}$ grows as $n^{2\varepsilon}$ at
most.

To compute $\text{Var}[\ztheta{1}]$ we use the following formulas
\begin{align}
  \label{eq:200}
  \text{Var}[\ztheta{1}] &= \sum_{k=n^\varepsilon + 1}^n
  \text{Var}[\zeta_{k\to k-1}]\\
  \text{Var}[\zeta_{k\to k-1}] &= \frac{2n}{k} \sum_{m=k}^n \left(
  2\mean{\zeta_{m\to k-1}} - \mean{\zeta_{k\to k-1}}\right)
\frac{\pin(m)}{\pin(k)}\nonumber\\ 
&\qquad\qquad - \mean{\zeta_{k\to k-1}} \label{eq:181}
\end{align}
Then we estimate the sum from below as its first term
\begin{align}
  \label{eq:201}
  \text{Var}[\zeta_{k\to k-1}]  \geq \left( \frac{2n}{k} -1 \right)
  \mean{\zeta_{k\to k-1}}
  = \left( \frac{2n}{k} -1 \right) \frac{2n}{k} \phi(k)
\end{align}
and from above as
\begin{align}
  \label{eq:203}
  \text{Var}[\zeta_{k\to k-1}] & \leq \frac{4n}{k} \sum_{m=k}^n
  \mean{\zeta_{m\to k-1}} \frac{\pin(m)}{\pin(k)}\\
  \label{eq:204}
  &\leq \frac{c\,n^2}{k} \sum_{m=k}^n \log\left( \frac{m}{k} \right)
  \frac{\pin(m)}{\pin(k)}\\
  \label{eq:205}
  &= c\,n^2 \sum_{m=k}^n \frac{2^{k-m}}{m} \log\left( \frac{m}{k} \right)\\
  \label{eq:206}
  &= c\,n^2 \sum_{j=0}^{n-k} \frac{2^{-j}}{k+j} \log\left( 1+\frac{j}{k} \right)\\
  \label{eq:207}
  &\leq \frac{c\,n^2}{k^2} \sum_{j=0}^{\infty} \, j \, 2^{-j}
\end{align}
From~\eqref{eq:201} and~\eqref{eq:207} we see that 
$\text{Var}[\zeta_{k\to k-1}] = O\left( \frac{n^2}{k^2}
\right)$ and therefore, to the leading order, $\text{Var}[\ztheta{1}] =
O\left( n^{2-\varepsilon} \right)$.

\end{appendices}


\phantomsection
\addcontentsline{toc}{section}{References}

\bibliography{lns}
\bibliographystyle{plainyr-rev}

\end{document}